\documentclass[lettersize,journal]{IEEEtran}

\usepackage{url}
\usepackage{xcolor}
\usepackage{xspace}
\usepackage{enumitem}
\usepackage{graphicx}
\usepackage{hyperref}
\usepackage{listings}
\usepackage{multirow}
\usepackage{subcaption}
\usepackage[normalem]{ulem}
\usepackage[most]{tcolorbox}
\usepackage{amsmath,amssymb,amsthm}
\usepackage[linesnumbered,ruled,vlined]{algorithm2e}

\def \tool{\textsc{FocusTNT}\xspace}
\def \base{\textsc{Base}\xspace}
\def \slice{\textsc{Slice}\xspace}
\def \concr{\textsc{Cncrt}\xspace}
\def \sliceconcr{\textsc{Slice+Cncrt}\xspace}
\def \concrslice{\textsc{Cncrt+Slice}\xspace}

\newtheorem{theorem}{Theorem}[section]
\newtheorem{corollary}[theorem]{Corollary}
\newtheorem{assumption}[theorem]{Assumption}

\newsavebox{\mybox}

\newtcolorbox{rqtakeaway}{
  colback=gray!10,
  colframe=gray!60,
  boxrule=0.5pt,
  arc=2mm,
  left=2mm,
  right=2mm,
  top=1mm,
  bottom=1mm
}

\lstdefinestyle{customcpp}{
  basicstyle=\ttfamily\footnotesize,
  columns=fullflexible,
  keepspaces=true,
  showstringspaces=false,
  frame=single,
  breaklines=true,
  tabsize=2,
  keywordstyle=\color{blue!60!black},
  numbers=left,
  numberstyle=\tiny,
  xleftmargin=2em,
  framexleftmargin=2em
}

\begin{document}

\title{Loop-Based Slicing and Input-Driven Concretization: An Empirical Study of Termination and Non-Termination Analysis}

\author{
Negar Fathi, Rahul Purandare, Tachio Terauchi, and Hiroshi Unno
\thanks{N. Fathi and R. Purandare are with the University of Nebraska--Lincoln, Lincoln, NE, USA. Email addresses: nfathi2@huskers.unl.edu, rahul@unl.edu.}
\thanks{T. Terauchi is with Waseda University, Tokyo, Japan. Email: terauchi@waseda.jp.}
\thanks{H. Unno is with Tohoku University, Sendai, Japan. Email: hiroshi.unno@acm.org.}
}

\maketitle

\begin{abstract}
Termination and non-termination are fundamental correctness properties, but verifying them in real-world C programs remains difficult because loop interactions and nondeterministic inputs challenge existing analyzers. This paper presents an empirical study of lightweight, tool-independent source-level preprocessing for (non-)termination analysis. We implement \tool, a C front end that applies loop-based slicing to isolate loop-level obligations and input-driven concretization to specialize nondeterministic inputs into selected input-scenario variants. We evaluate slicing, concretization, and their combination across six analyzers on 117 C/C++ programs derived from real-world non-termination bugs and their fixes. The study examines effects on analyzer correctness, complementarity with original-program analysis, loop-level diagnostics, feature sensitivity, runtime behavior, semantic scope, and integration potential. Results show that preprocessing is not uniformly beneficial: its impact depends on the analyzer, task, and program features. Slicing provides conservative structural isolation and localization, whereas concretization can improve detectability for selected scenarios but narrows semantic scope and may increase analysis effort. Their combination is not consistently additive. Overall, the results support adaptive use of preprocessing as a complement to original-program analysis and provide practical guidance to application developers interpreting verification outcomes and tool developers improving analyzer robustness.
\end{abstract}


\section{Introduction}
\label{section:introduction}
Termination and non-termination are fundamental correctness properties in program analysis, especially for systems software where infinite execution can compromise responsiveness and reliability~\cite{cook2006,xie2017}. Failure to establish termination can also obstruct verification tasks that assume eventual return, such as liveness reasoning and compositional arguments~\cite{cook2006}. In practice, proving (non-)termination remains difficult because loops are often embedded in large, semantically noisy contexts that obscure the statements and dependencies most relevant to termination behavior~\cite{cook2006,xie2017}. Environment-driven nondeterminism further complicates analysis by exposing many input-dependent behaviors, only some of which may matter for a particular (non-)termination outcome~\cite{chen2014,larraz2014,xie2017}.

Static analyzers address termination properties by constructing proofs that generalize across executions, combining abstraction~\cite{urban2013,bakhirkin2015}, invariant inference~\cite{colon2002,bradley2005linear,kroening2010}, ranking-function synthesis~\cite{colon2001,podelski2004,bradley2005linear,leike2014,benamram2014}, recurrence reasoning~\cite{gupta2008,cook2014,larraz2014,bakhirkin2015}, and solver-based back ends~\cite{larraz2014,chen2014,kroening2010}. Although such tools perform well on established benchmarks such as TermCOMP~\cite{termcomp} and SV-COMP~\cite{svcomp}, prior studies show that they still struggle on programs derived from real-world non-termination bugs, where complex program context, nondeterministic inputs, and C/C++-specific features make proof construction difficult~\cite{shi2022,cook2011}. Dynamic approaches complement static analysis by using concrete executions to expose input-specific (non-)termination behaviors and generate witnesses when divergence depends on particular input patterns~\cite{le2020,karmarkar2022,milicevic2023}. However, execution evidence alone cannot establish general termination guarantees, since observing termination on finitely many executions does not imply termination of all executions~\cite{le2020,milicevic2023}.

This paper studies (non-)termination analysis from a preprocessing perspective: rather than proposing a new termination prover, we ask whether lightweight source-level transformations can reshape verification tasks before they reach existing analyzers. We focus on two transformations motivated by the challenges above: loop-based slicing, which isolates loop-relevant structure while preserving the analyzed loop’s termination behavior, and input-driven concretization, which specializes nondeterministic inputs into selected input-scenario variants. These transformations may help analyzers focus on loop-relevant dependencies or reason about restricted input-dependent behaviors, but they may also remove useful context, expose different behaviors, or narrow semantic scope; their effects therefore require empirical evaluation.

We implement this perspective in \tool, a lightweight, tool-independent preprocessing front end for C programs that supports loop-based slicing and input-driven concretization. We evaluate four configurations—\base (no preprocessing), \slice (slicing only), \concr (concretization only), and \sliceconcr (slicing followed by concretization)—across six analyzers: Athena~\cite{athena}, Proton~\cite{metta2024,mukhopadhyay2025}, UAutomizer~\cite{heizmann2014}, AProVE~\cite{giesl2014}, CPAchecker~\cite{beyer2011}, and 2LS~\cite{schrammel2016}. Our study uses 117 C/C++ programs derived from real-world non-termination bugs and their fixes~\cite{shi2022,fse} and examines analyzer correctness, complementarity with original-program analysis, loop-level diagnostics, feature sensitivity, runtime behavior, semantic scope, and integration potential. These dimensions provide practical insight for application developers interpreting preprocessing-based verification outcomes and tool developers improving analyzer robustness.

In summary, this work makes the following contributions:
\begin{enumerate}[leftmargin=*,noitemsep,topsep=0pt]
    \item We design and implement \tool, a lightweight, tool-independent C preprocessing framework that supports loop-based slicing, input-driven concretization, and their combination without modifying backend analyzers.
    \item We conduct a multi-tool empirical evaluation of four configurations---\base, \slice, \concr, and \sliceconcr---across six termination and non-termination analyzers on 117 C/C++ programs derived from real-world non-termination bugs and their fixes.
    \item We assess preprocessing effects on analyzer correctness, complementarity with original-program analysis, loop-level diagnostics, feature sensitivity, runtime behavior, semantic scope, and integration potential.
    \item We derive practical implications for adaptive preprocessing, helping application developers interpret preprocessing-based verification outcomes and tool developers identify feature- and transformation-sensitive cases for improving analyzer robustness.
\end{enumerate}

The remainder of the paper is organized as follows. Section~\ref{section:motivating-examples} presents motivating examples. Section~\ref{section:related-work} reviews related work. Sections~\ref{section:approach} and~\ref{section:implementation} introduce the preprocessing framework and implementation. Sections~\ref{section:experimental-setup} and~\ref{section:results} define the experimental setup and present the empirical results. Section~\ref{section:discussion} discusses the main findings and implications, Section~\ref{section:threats} examines threats to validity, and Section~\ref{section:conclusion} concludes with future work.

\section{Motivating Examples}
\label{section:motivating-examples}
Shi et al.~\cite{shi2022,fse} introduced a benchmark of simplified C/C++ programs derived from real-world non-termination bugs and their fixes. We use two representative cases to illustrate how loop-based slicing and input-driven concretization reshape verification tasks: \texttt{Misusing\_Variable\_Type\_1\_NT} (Figure~\ref{figure:slice}) and \texttt{Incorrect\_Control\_Statement\_2\_NT} (Figure~\ref{figure:cncrt}).

\begin{figure}[!t]
\centering
\begin{subfigure}[t]{0.98\columnwidth}
\centering
\begin{lstlisting}[style=customcpp, firstnumber=1]
int main() {
  unsigned int mul, div1, div2;
  for(div1 = 1; div1 >= 0; div1--) {
    for(div2 = 7; div2 >= 0; div2--) {
      for(mul = 0; mul <= 255; mul++) {
  }}}
  return 0;
}
\end{lstlisting}
\vspace{-8pt}
\caption{Original program.}
\vspace{4pt}
\label{figure:original}
\end{subfigure}
\begin{subfigure}[t]{0.98\columnwidth}
\centering
\begin{lstlisting}[style=customcpp, firstnumber=1]
void main(void) {
  unsigned int div1;
  div1 = (unsigned int)1;
  while (div1 >= (unsigned int)0) div1--;
}
\end{lstlisting}
\vspace{-8pt}
\caption{Outer-loop slice.}
\vspace{4pt}
\label{figure:outer}
\end{subfigure}
\begin{subfigure}[t]{0.98\columnwidth}
\centering
\begin{lstlisting}[style=customcpp, firstnumber=1]
void main(void) {
  unsigned int div2;
  div2 = (unsigned int)7;
  while (div2 >= (unsigned int)0) div2--;
}
\end{lstlisting}
\vspace{-8pt}
\caption{Middle-loop slice.}
\vspace{4pt}
\label{figure:middle}
\end{subfigure}
\begin{subfigure}[t]{0.98\columnwidth}
\centering
\begin{lstlisting}[style=customcpp, firstnumber=1]
void main(void) {
  unsigned int mul;
  mul = (unsigned int)0;
  while (mul <= (unsigned int)255) mul++;
}
\end{lstlisting}
\vspace{-8pt}
\caption{Inner-loop slice.}
\label{figure:inner}
\end{subfigure}
\caption{Motivating example for loop-based slicing: \texttt{Misusing\_Variable\_Type\_1\_NT}, from Shi et al.'s benchmark suite~\cite{shi2022,fse}.}
\label{figure:slice}
\end{figure}

Figure~\ref{figure:original} shows a program with three nested unsigned-integer loops. The two decrementing loops are non-terminating because unsigned wrap-around keeps their guards satisfiable, while the innermost loop is bounded and terminating. Under \base, analyzers must reason about one program coupling all three loop variables and control structure: only Athena and 2LS prove non-termination, UAutomizer, AProVE, and CPAchecker are inconclusive, and Proton incorrectly reports termination. Loop-based slicing (Section~\ref{subsection:slicing}) decomposes the program into loop-centric variants (Figures~\ref{figure:outer}, \ref{figure:middle}, and \ref{figure:inner}), each preserving the termination behavior relevant to one loop. This isolates loop-local reasoning tasks and reduces interference from surrounding loops. On these slices, previously inconclusive analyzers prove termination of the bounded inner loop and non-termination of at least one diverging loop, while Proton no longer reports termination and instead becomes inconclusive. Under the aggregation policy in Algorithm~\ref{algorithm:tool-decision-evaluation}, detecting one non-terminating slice suffices to classify the original non-terminating benchmark correctly. Thus, all analyzers except Proton reach a correct program-level conclusion under slicing, and Proton's earlier incorrect termination result is eliminated. This example shows how slicing localizes divergence reasoning and exposes loop-specific behavior obscured in the original program.

\begin{figure}[!t]
\centering
\begin{lstlisting}[style=customcpp, firstnumber=1, linewidth=0.98\columnwidth]
extern int __VERIFIER_nondet_int(void);
#define EVEBT_EPOLL_SLOTS 2
int main()  {
  int EVENT_EPOLL_TABLES = 10;
  int slots_used[10];
  int ereg[10];
  int table;
  for (int i = 0; i < 10; i++)  {
    slots_used[i] = __VERIFIER_nondet_int();
    ereg[i] = __VERIFIER_nondet_int();
  }
  int i = 0;
  while (i < EVENT_EPOLL_TABLES)  {
    switch (slots_used[i]) {
      case EVEBT_EPOLL_SLOTS:
        continue;
      case 0:
        if (!ereg[i]) return 0;
        else table = ereg[i];
        break;
      default:
        table = ereg[i];
        break;
    }
    if (table) break;
    i++;
  }
  return 0;
}
\end{lstlisting}
\caption{Motivating example for input-driven concretization: \texttt{Incorrect\_Control\_Statement\_2\_NT}, from Shi et al.'s benchmark suite~\cite{shi2022,fse}.}
\label{figure:cncrt}
\end{figure}

Figure~\ref{figure:cncrt} shows a program whose loop behavior depends on arrays initialized with \texttt{\_\_VERIFIER\_nondet\_int()}. The counter \texttt{i} advances only on some control-flow paths: if \texttt{slots\_used[i]} repeatedly takes a particular value, execution bypasses the increment of \texttt{i}, preventing progress and causing non-termination; otherwise, \texttt{i} eventually increases and the loop terminates. Under \base, analyzers must reason about both terminating and diverging behaviors induced by nondeterministic inputs: Proton and UAutomizer prove non-termination, AProVE and 2LS are inconclusive, Athena times out, and CPAchecker reports an error. Input-driven concretization (Section~\ref{subsection:concretization}) specializes nondeterministic inputs by instantiating \texttt{slots\_used} and \texttt{ereg} with concrete values derived from candidate inputs discovered and validated by the procedure described in that section.
The resulting variants restrict analysis to selected input-scenario behaviors, reducing nondeterministic variability and making input-specific termination or divergence easier to expose. After concretization, Athena and 2LS, which did not resolve the original program, resolve several specialized variants, including some classified as non-terminating. Under the aggregation policy in Algorithm~\ref{algorithm:tool-decision-evaluation}, the non-terminating variant results suffice for a correct overall non-termination classification. This example shows how input-driven concretization complements static reasoning by presenting analyzers with restricted execution domains.

Together, these examples show that preprocessing can reshape verification tasks without changing backend analyzers. Loop-based slicing isolates loop structure while preserving its termination behavior, whereas input-driven concretization restricts nondeterministic variability to selected input scenarios. These transformations expose alternative verification views with distinct semantic effects, motivating our empirical study of their benefits, limitations, and interactions across analyzers.


\section{Related Work}
\label{section:related-work}
\begin{figure*}[!t]
    \centering
    \includegraphics[width=0.9\textwidth]{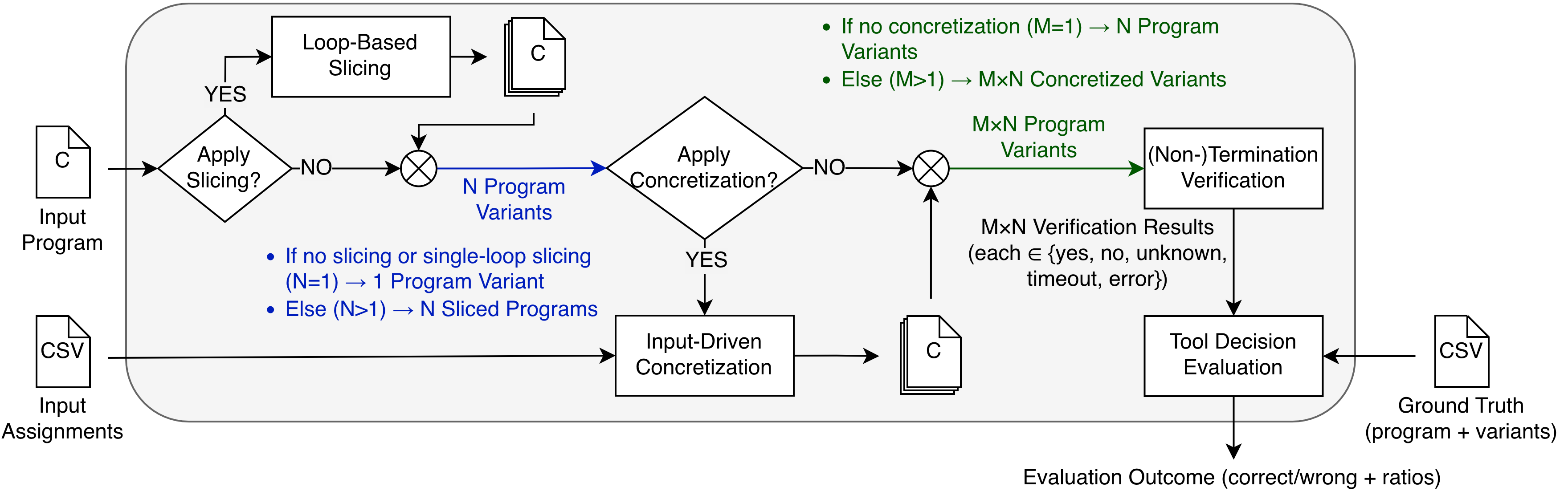}
    \caption{Overview of the \tool preprocessing and verification pipeline.}
    \label{figure:pipeline}
\end{figure*}

We position our work with respect to program slicing and static, dynamic, and hybrid approaches to (non-)termination analysis.

a) {\bf Program Slicing:}
\label{subsection:rw-slicing}
Program slicing extracts components relevant to a slicing criterion, such as selected variable values at a program point~\cite{weiser1981}. Static slicing over-approximates behavior across executions, whereas dynamic slicing is tied to a concrete execution and produces a trace-specific slice that may not generalize beyond the observed run~\cite{weiser1981,bogdan1988}. Several infrastructures support C/C++ slicing at the source or intermediate representation (IR) level: CodeSurfer~\cite{anderson2005} builds system dependence graphs for forward and backward slicing and chopping queries; Frama-C~\cite{cuoq2012} provides source-level slicing integrated with abstract interpretation; DG~\cite{Chalupa2020} performs dependence-graph-based slicing over LLVM bitcode; and Giri~\cite{sahoo2013} performs dynamic backward slicing from execution traces.

Prior work has used slicing mainly for debugging, program comprehension, and verification reduction. In contrast, we use loop-based slicing as a verification-oriented structural isolation step for (non-)termination analysis. Our C source-level implementation uses Frama-C~\cite{cuoq2012}, avoiding IR-level translation and potential decompilation issues while remaining compatible with existing termination analyzers. 

b) {\bf Static (Non-)Termination Analysis:}
\label{subsection:rw-static-tnt}
Termination and non-termination analysis for imperative programs is predominantly static and proof-oriented, aiming to establish guarantees that generalize across executions. We briefly describe the state-of-the-art C analyzers used in our study, which represent different verification paradigms.

Athena~\cite{athena} is a sound C (non-)termination analyzer that models low-level semantics using pointer-to-array rewriting and bit-precise bounded-integer reasoning. It translates programs into logical transition systems and applies $\mu$CLP solving in the MuVal primal--dual fixpoint framework~\cite{unno2021,kura2021,unno2023}, synthesizing ranking functions and recurrent sets through mutually recursive inductive and co-inductive predicates. Proton~\cite{metta2024,mukhopadhyay2025} detects non-termination by finding recurrent program states with bounded model checking and supports termination analysis through LLM-assisted candidate ranking functions validated by randomized testing and bounded model checking. UAutomizer~\cite{heizmann2013,heizmann2014,heizmann2015,heizmann2016,heizmann2018} is an automata-based verifier using trace abstraction and interpolation-based refinement, with (non-)termination reasoning over infinite traces and lasso-shaped executions. AProVE~\cite{giesl2014,falke2014,hensel2016,noll2018} builds symbolic execution graphs that over-approximate executions and translates them into integer transition or rewrite systems for automated (non-)termination analysis. CPAchecker~\cite{beyer2011} is a configurable verification framework that performs reachability analysis over control-flow automata by combining abstract domains such as predicate abstraction and explicit-value analysis. 2LS~\cite{schrammel2016,kaiser2017,kaiser2018,chen2022} is a bit-precise CPROVER-based analyzer combining bounded model checking, k-induction, and template-based invariant synthesis for interprocedural termination and non-termination analysis over bitvector semantics.

Our work is complementary to these analyzers: rather than modifying their reasoning engines or adding proof rules, we treat them as black boxes and study how lightweight source-level preprocessing reshapes their verification tasks. This enables an empirical analysis of when preprocessing helps or hurts, and how its effects depend on analyzer architecture, program features, and the termination versus non-termination task. The findings help application developers interpret preprocessing-based outcomes and tool developers improve analyzer robustness.

c) {\bf Dynamic and Hybrid (Non-)Termination Analyses:}
\label{subsection:rw-dynamic-hybrid-tnt}
Dynamic and hybrid approaches complement static reasoning by using concrete executions to expose or explain termination behavior. DynamiTe~\cite{le2020} collects traces, learns candidate ranking functions and recurrent sets, and validates them with SMT-based static reasoning, using failed validations to drive further executions in a dynamic--static refinement loop. FuzzNT~\cite{karmarkar2022} combines coverage-guided fuzzing with a guess-and-check workflow that builds path-specialized under-approximating variants, analyzed by abstract interpretation to confirm non-termination. EndWatch~\cite{milicevic2023} targets real-world software by instrumenting loops with state-revisit non-termination oracles and exploring executions via fuzzing and symbolic execution.

Our input-driven concretization is related to these hybrid approaches because it uses concrete input scenarios to specialize analysis tasks. However, rather than building a new dynamic--static verification engine, we externalize specialization as source-level preprocessing and evaluate how the resulting variants affect existing analyzers. Combined with loop-based slicing, this lets us study structural isolation and input-scenario specialization as complementary preprocessing dimensions for (non-)termination analysis and guides their integration into verification workflows. 

\section{Approach}
\label{section:approach}
This section presents \tool, the source-level preprocessing pipeline underlying our empirical study.

\subsection{Pipeline Overview}
\label{subsection:overview}
Figure~\ref{figure:pipeline} gives an overview of the \tool preprocessing pipeline. The pipeline takes as input a C program and a tabular input specification whose columns identify nondeterministic input locations and whose rows define concrete input assignments. The first stage optionally applies loop-based slicing and produces $N$ program variants. If slicing is disabled, this stage produces the original program ($N=1$); if the program has a single loop, it produces one sliced variant ($N=1$); and if the program has multiple loops, it produces $N>1$ loop-focused sliced variants. The second stage optionally applies input-driven concretization. If concretization is disabled, the current $N$ variants are analyzed unchanged ($M=1$); otherwise, each variant is specialized using $M>1$ selected input assignments, yielding $M \times N$ concretized variants. Each generated variant is analyzed independently by off-the-shelf termination and non-termination tools, which return one of {\emph{yes}, \emph{no}, \emph{unknown}, \emph{timeout}, \emph{error}}. Algorithm~\ref{algorithm:tool-decision-evaluation} evaluates the resulting outcomes against program- and variant-level ground truth, producing a program-level decision (\emph{correct} or \emph{wrong}) and the fractions of terminating and non-terminating variants correctly classified, denoted $ratio_T$ and $ratio_{NT}$.

\subsection{Loop-Based Slicing}
\label{subsection:slicing}
The loop-based slicing phase isolates each syntactic loop in a separate program variant. Given a program $P$, \tool generates one slice per loop, retaining the code required by the slicing criterion to preserve termination-relevant behavior for that loop while removing unrelated variables, statements, and functions when possible. This gives analyzers a smaller, more focused verification task with less interference from surrounding context. The resulting slices are loop-focused variants interpreted relative to the selected loop and slicing criterion, not claims of whole-program semantic equivalence. In particular, the current slicing procedure targets syntactic loop obligations rather than recursive call cycles. Accordingly, it does not introduce separate slicing criteria or recursion-focused variants for recursive code; such code is considered only insofar as it contributes to the context of a selected loop and the dependences captured by its slicing criterion.

For each \texttt{for}, \texttt{while}, or \texttt{do}--\texttt{while} loop $\ell$ in $P$, \tool uses the Clang~\cite{llvmclang} abstract syntax tree to construct a loop descriptor
\[
D(\ell) =
    \bigl(
        \mathit{loc}(\ell),
        \mathit{func}(\ell),
        T_{\mathit{cond}}(\ell),
        S_{\mathit{ctrl}}(\ell)
    \bigr),
\]
where:
\begin{itemize}[leftmargin=*,noitemsep,topsep=0pt]
    \item $\mathit{loc}(\ell)$ is the source location of $\ell$.
    \item $\mathit{func}(\ell)$ is the function enclosing $\ell$.
    \item $T_{\mathit{cond}}(\ell)$ is the set of lvalues in the loop guard. An lvalue denotes a memory location, such as a scalar variable, array element, or field access. Scalars are recorded directly, while array and field accesses are abstracted to their base memory objects to conservatively capture dependences.
    \item $S_{\mathit{ctrl}}(\ell)$ is the set of source locations of loop-control statements in the loop body, including \texttt{continue}, \texttt{break}, \texttt{goto}, and \texttt{return}.
\end{itemize}

The descriptor induces the loop-local slicing criterion
\[
C(\ell) =
    \operatorname{Read} \! \left( T_{\mathit{cond}}(\ell) \right)
    \cup
    \operatorname{Annotate} \! \left( S_{\mathit{ctrl}}(\ell) \right),
\]
where $\operatorname{Read}(\cdot)$ denotes reads of guard-relevant lvalues and $\operatorname{Annotate}(\cdot)$ denotes ANSI/ISO C Specification Language (ACSL) annotations inserted at the recorded loop-control locations. Operationally, \tool encodes this criterion as Frama-C~\cite{cuoq2012} slicing directives over the selected guard reads and annotated loop-control locations. Frama-C then retains statements needed to preserve the associated data and control dependences.

Loops in \texttt{main} and auxiliary functions are handled differently. If $\ell$ appears in \texttt{main}, \tool invokes Frama-C with \texttt{main} as the entry point and slices directly with respect to $C(\ell)$, producing $\mathsf{Slice}(P,\ell)$. For loops in non-\texttt{main} functions, \tool constructs two slices. The callee-focused slice is rooted at $\mathit{func}(\ell)$ and uses the criterion
\[
C(\ell) \cup \operatorname{KeepReturn} \! \left( \mathit{func}(\ell) \right),
\]
where $\operatorname{KeepReturn}(\cdot)$ is the Frama-C directive used to retain the return behavior of the enclosing function. The caller-focused slice is rooted at \texttt{main} and uses
\[
\operatorname{KeepCalls} \! \left( \mathit{func}(\ell) \right),
\]
where $\operatorname{KeepCalls}(\cdot)$ retains calls to $\mathit{func}(\ell)$ along the invocation context rooted at \texttt{main}. The final slice $\mathsf{Slice}(P,\ell)$ replaces the corresponding function body in the caller-focused slice with the callee-focused reduced version, retaining the relevant invocation context while keeping the generated program focused on $\ell$.

\subsection{Input-Driven Concretization}
\label{subsection:concretization}
The input-driven concretization phase specializes nondeterministic inputs for selected input scenarios. Given a program variant from the previous stage and a tabular input specification, it produces one concretized variant per input assignment by replacing nondeterministic input generators with type-correct concrete values. By reducing nondeterministic variability, concretization can make input-dependent termination or divergence easier to expose. The resulting variants represent selected input-scenario behaviors and therefore do not cover all executions of the original program.

Concrete inputs are specified as a table whose header contains $N$ identifiers $h_1, h_2, \dots, h_N$ and whose $M$ rows define individual input assignments. Each identifier has the form
\[
h_i = x_i :: f_i :: k_i
\qquad
1 \le i \le N,
\]
where $x_i$ is the assigned lvalue expression when a corresponding \texttt{\_\_VERIFIER\_nondet\_*()} call appears on the right-hand side of an initialization or assignment, and the reserved token \texttt{NONDET} otherwise; $f_i$ is the enclosing function; and $k_i \in \mathbb{N}$ distinguishes syntactic occurrences within the function by source order. Each row defines an assignment
\[
t_j = \{ h_1 \mapsto c_{j1}, h_2 \mapsto c_{j2}, \dots, h_N \mapsto c_{jN} \}
\qquad
1 \le j \le M,
\]
where each $c_{ji}$ is a type-correct concrete literal assigned to $h_i$. The resulting input set is denoted $\mathcal{T} = { t_1, t_2, \dots, t_M }$.

Given a program variant $P$ and input assignment $t \in \mathcal{T}$, the concretized program $\mathsf{Cncrt}(P,t)$ specializes every nondeterministic call in $P$ according to $t$. Each call is mapped to its identifier $h = x :: f :: k$ and rewritten by replacing the call expression with the literal $c = t(h)$ while preserving the surrounding expression context. We assume the input specification assigns values to all identifiers, so nondeterministic inputs are fully specialized. In the current implementation, repeated evaluations of the same syntactic nondeterministic call within one execution, including inside loops, use the same concrete value.

Concretization is an intentional under-approximation: $\mathsf{Cncrt}(P,t)$ restricts $P$ to executions consistent with a selected input assignment $t$, rather than representing all nondeterministic choices. Thus, termination of $\mathsf{Cncrt}(P,t)$ holds only for that input scenario and does not imply universal termination of $P$. Conversely, if $\mathsf{Cncrt}(P,t)$ is non-terminating and $t$ is feasible for $P$, it witnesses non-termination of the original program. Concretization can therefore expose input-dependent behavior more effectively, but its conclusions must be interpreted with respect to the selected assignments. We implement concretization as a Clang-based source-to-source transformation that generates one variant for each pair $(P,t)$; if concretization is disabled or no input specification is provided, the program is analyzed unchanged.

Input generation is not a contribution of this work. Input assignments are used only to instantiate selected scenarios for evaluating input-driven concretization. We construct input sets through controlled candidate discovery followed by manual validation. For each benchmark, OpenAI's ChatGPT (GPT-4o) proposes candidate assignments, which are then curated by discarding inconsistent cases and refining execution scenarios into a fixed set of feasible assignments, including both terminating and diverging scenarios when available. This LLM-assisted step is only a practical mechanism for obtaining diverse assignments: the model is not treated as an oracle and is not used to establish ground truth. We also evaluated automated generators such as KLEE~\cite{klee} and AFL~\cite{afl}, but complex C features and timeout-truncated diverging executions often limited the usefulness of automatically derived witnesses in our setting.

\subsection{Termination and Non-Termination Verification}
\label{subsection:verification}
\begin{algorithm*}[!t]
  \caption{Program-Level Tool Decision Evaluation}
  \label{algorithm:tool-decision-evaluation}

  \KwIn{
    Program-level ground truth $GT_P \in \{ \mathrm{T}, \mathrm{NT} \}$; variant-level ground truths $(GT_{V_i})_{1 \le i \le n}$ with $GT_{V_i} \in \{ \mathrm{T}, \mathrm{NT} \}$; and variant-level analyzer outputs $(O_{V_i})_{1 \le i \le n}$ with $O_{V_i} \in \{ \texttt{yes}, \texttt{no}, \texttt{unknown}, \texttt{timeout}, \texttt{error} \}$.
  }

  \KwOut{
    A program-level decision $d \in \{ \texttt{correct}, \texttt{wrong} \}$ and variant-resolution ratios: $(d, ratio_T)$ if $GT_P = \mathrm{T}$, and $(d, ratio_T, ratio_{NT})$ if $GT_P = \mathrm{NT}$.
  }

  \If{$GT_P=\mathrm{T}$}{
    \lIf{$\exists\, i\ (1 \le i \le n).\ O_{V_i}=\texttt{no}$}{
      \Return{$(\texttt{wrong},\bot)$}
    }
    $ratio_T \leftarrow |\{\, i \mid 1 \le i \le n \land O_{V_i}=\texttt{yes} \,\}|\ /\ n$\;
    \Return{$(\texttt{correct},\, ratio_T)$}\;
  }
  \Else{
    \lIf{$\exists\, i\ (1\le i\le n).\ (GT_{V_i}=\mathrm{T}\land O_{V_i}=\texttt{no})\ \lor\ (GT_{V_i}=\mathrm{NT}\land O_{V_i}=\texttt{yes})$}{
      \Return{$(\texttt{wrong},\bot,\bot)$}
    }
    $ratio_T \leftarrow |\{\, i \mid 1 \le i \le n \land GT_{V_i}=\mathrm{T}\land O_{V_i}=\texttt{yes} \,\}|\ /\ |\{\, i \mid 1 \le i \le n \land GT_{V_i}=\mathrm{T} \,\}|$\;
    $ratio_{NT} \leftarrow |\{\, i \mid 1 \le i \le n \land GT_{V_i}=\mathrm{NT}\land O_{V_i}=\texttt{no} \,\}|\ /\ |\{\, i \mid 1 \le i \le n \land GT_{V_i}=\mathrm{NT} \,\}|$\;
    \Return{$(\texttt{correct},\, ratio_T,\, ratio_{NT})$}\;
  }
\end{algorithm*}

We evaluate four analysis configurations that differ only in preprocessing before verification: \base (no preprocessing), \slice (loop-based slicing only), \concr (input-driven concretization only), and \sliceconcr (slicing followed by concretization). These configurations isolate the individual and combined effects of loop-level structural isolation and input-scenario specialization while keeping backend analyzers and the verification workflow unchanged. The combined configuration, \sliceconcr, is included to study interactions between slicing and concretization, not as a prescribed deployment strategy.

Under each configuration, the pipeline generates one or more program variants from the same original program. Each variant is analyzed independently by the selected off-the-shelf termination and non-termination tools. Each analyzer invocation returns one of {\textit{yes}, \textit{no}, \textit{unknown}, \textit{timeout}, \textit{error}}: \textit{yes} denotes a termination proof, \textit{no} a non-termination proof, \textit{unknown} an inconclusive result within the time limit, \textit{timeout} time-limit exceedance, and \textit{error} an internal tool failure. This variant-level workflow lets us evaluate preprocessing effects while treating analyzers as black boxes.

For each tool and configuration, Algorithm~\ref{algorithm:tool-decision-evaluation} aggregates variant-level outputs into a program-level result using the ground truth of both the original program and generated variants. For terminating programs, the result is marked \emph{wrong} if any variant is classified as non-terminating; otherwise, performance is measured by the fraction of variants proven terminating, denoted $ratio_T$. For non-terminating programs, the result is marked \emph{wrong} if any variant-level output contradicts the corresponding variant ground truth; otherwise, performance is measured by the fractions of terminating and non-terminating variants correctly classified, denoted $ratio_T$ and $ratio_{NT}$. Together, these metrics capture program-level correctness and variant-level coverage, enabling principled comparison across tools and preprocessing configurations.

\subsection{Semantic Scope and Preservation Guarantees}
\label{subsection:semantic-preservation}
We formalize the semantic scope of the transformations and their composition in terms of preserved termination-relevant behaviors rather than full program equivalence. The loop-focused notation and preservation claims range over syntactic loops and exclude non-termination caused solely by recursive call cycles without infinite iteration of such constructs.

\paragraph*{Preliminaries}
Let $P$ be a program, $\ell \in \mathit{Loops}(P)$ a syntactic loop, and $\Sigma(P)$ the set of well-defined initial states of $P$. For slicing, let $C(\ell)$ be the loop-local slicing criterion from Section~\ref{subsection:slicing}. Because slicing may remove context that affects reachability of $\ell$ or the states reaching it, we consider only states that are well-defined in both $P$ and its slice and from which $\ell$ is reachable in $P$:
\[
\begin{aligned}
    \Sigma(P,\ell) \triangleq
    \{ \sigma \in \Sigma(P) \cap \Sigma(\mathsf{Slice}(P,\ell)) \mid {} &
    \ell \text{ is reachable in } P \\
    & \text{ from } \sigma \}.
\end{aligned}
\]
We also define a reachability-precedence relation on loops. For loops $\ell',\ell \in \mathit{Loops}(P)$ and state $\sigma$, we write
\[
\ell' \preceq^P_\sigma \ell
\]
if $\ell'=\ell$ or some execution of $P$ from $\sigma$ encounters $\ell'$ before the first encounter of $\ell$. For $\sigma \in \Sigma(P,\ell)$, we write $P,\sigma \Downarrow_\ell$ if every well-defined execution from $\sigma$ performs finitely many iterations of $\ell$, and $P,\sigma \Uparrow_\ell$ if some such execution performs infinitely many iterations. For concretization, let $\mathcal{T}$ be the selected input assignments, each complete and type-correct for the nondeterministic input occurrences in the input specification. For $t \in \mathcal{T}$, $P[t]$ denotes the semantic restriction of $P$ where each specified \texttt{\_\_VERIFIER\_nondet\_*()} outcome is fixed according to $t$, with repeated evaluations of the same syntactic occurrence fixed consistently to the same value. Thus, $P[t]$ represents one selected input scenario. For $\sigma \in \Sigma(P)$, we write $P[t],\sigma \Downarrow$ and $P[t],\sigma \Uparrow$ for termination and divergence under $t$. For combined slicing and concretization, $P[t],\sigma \Downarrow_\ell$ and $P[t],\sigma \Uparrow_\ell$ denote finite and infinite iteration of $\ell$, respectively, under $t$.

\begin{assumption}[Loop-criterion adequacy of generated slices]
\label{assumption:slicing-adequacy}
    For each generated slice $\mathsf{Slice}(P,\ell)$, we assume that the slicing backend preserves the loop-local criterion $C(\ell)$, which comprises the guard-relevant reads of $\ell$, the annotated loop-control locations inside $\ell$, and their data and control dependencies. Thus, executions from states in $\Sigma(P,\ell)$ that reach $\ell$ have corresponding sliced executions with the same retained criterion behavior, and any sliced loop-divergence relevant here arises from behavior retained from $P$, either infinite iteration of $\ell$ itself or divergence of a loop that can be encountered before $\ell$ in $P$. The preservation results below are conditional loop-focused consequences of this adequacy, not claims of whole-program semantic equivalence.
\end{assumption}

\begin{theorem}[Loop-focused non-termination preservation under slicing]
\label{theorem:slicing-nontermination}
Under Assumption~\ref{assumption:slicing-adequacy}, for every program
$P$, loop $\ell \in \mathit{Loops}(P)$, and state
$\sigma \in \Sigma(P,\ell)$:
    \begin{enumerate}[leftmargin=*,noitemsep,topsep=0pt]
        \item $P,\sigma \Uparrow_\ell \Longrightarrow \mathsf{Slice}(P,\ell),\sigma \Uparrow_\ell$;
        \item $\mathsf{Slice}(P,\ell),\sigma \Uparrow_\ell \Longrightarrow \exists \ell' \in \mathit{Loops}(P).\; \ell' \preceq^P_\sigma \ell \wedge \sigma \in \Sigma(P,\ell') \wedge P,\sigma \Uparrow_{\ell'}$.
    \end{enumerate}
\end{theorem}

\begin{proof}[Proof sketch]
    For (1), Assumption~\ref{assumption:slicing-adequacy} ensures that any execution of $P$ from $\sigma$ witnessing infinite iteration of $\ell$ has a corresponding execution in $\mathsf{Slice}(P,\ell)$ with the same retained behavior relevant to $C(\ell)$. Hence, the slice preserves non-termination of the selected loop obligation. For (2), assume that $\mathsf{Slice}(P,\ell)$ has an infinite execution of $\ell$ from $\sigma$. By Assumption~\ref{assumption:slicing-adequacy}, this execution is justified by criterion behavior retained from $P$. If the corresponding original execution reaches $\ell$, the divergence is the selected loop divergence itself. Otherwise, the obstruction occurs earlier, in some loop $\ell'$ encountered before the first encounter of $\ell$ from $\sigma$. Thus, the sliced divergence is accounted for by an original-program divergence in $\ell$ or in a loop preceding $\ell$ from $\sigma$.
\end{proof}

\begin{theorem}[Loop-focused termination preservation under slicing]
\label{theorem:slicing-termination}
    Under Assumption~\ref{assumption:slicing-adequacy}, for every program $P$, loop $\ell \in \mathit{Loops}(P)$, and state $\sigma \in \Sigma(P,\ell)$:
    \begin{enumerate}[leftmargin=*,noitemsep,topsep=0pt]
        \item $\mathsf{Slice}(P,\ell),\sigma \Downarrow_\ell \Longrightarrow P,\sigma \Downarrow_\ell$;
        \item $\bigl(\forall \ell' \in \mathit{Loops}(P).\; \ell' \preceq^P_\sigma \ell \wedge \sigma \in \Sigma(P,\ell') \Rightarrow P,\sigma \Downarrow_{\ell'}\bigr) \Longrightarrow \mathsf{Slice}(P,\ell),\sigma \Downarrow_\ell$.
    \end{enumerate}
\end{theorem}

\begin{proof}[Proof sketch]
    For (1), we argue by contraposition using Theorem~\ref{theorem:slicing-nontermination}(1). Any infinite iteration of $\ell$ in $P$ would be preserved as divergence of the corresponding loop-focused obligation in $\mathsf{Slice}(P,\ell)$. Hence, termination of that obligation in the slice implies termination of $\ell$ in $P$. For (2), contraposition with Theorem~\ref{theorem:slicing-nontermination}(2) shows that any infinite loop-focused execution in the slice would be accounted for by divergence of some original-program loop preceding $\ell$ from $\sigma$. Since the premise rules out all such divergences, the selected loop-focused obligation in the slice must terminate. These are loop-focused claims, not whole-program equivalence claims.
\end{proof}

\begin{theorem}[Scenario-level correctness of concretization]
\label{theorem:concretization}
    For every program $P$, complete and type-correct input assignment $t \in \mathcal{T}$, and initial state $\sigma \in \Sigma(P)$:
    \begin{enumerate}[leftmargin=*,noitemsep,topsep=0pt]
        \item $P[t],\sigma \Uparrow \Longleftrightarrow \mathsf{Cncrt}(P,t),\sigma \Uparrow$;
        \item $P[t],\sigma \Downarrow \Longleftrightarrow \mathsf{Cncrt}(P,t),\sigma \Downarrow$.
    \end{enumerate}
\end{theorem}

\begin{proof}[Proof sketch]
    The semantic program $P[t]$ is $P$ restricted to the selected input assignment $t$. The transformation $\mathsf{Cncrt}(P,t)$ implements this restriction syntactically by replacing each specified \texttt{\_\_VERIFIER\_nondet\_*()} occurrence with the corresponding type-correct literal from $t$, while preserving declarations, expression contexts, control flow, and evaluation order. Because $t$ fixes all specified nondeterministic inputs, including repeated evaluations of the same syntactic occurrence, executions of $P[t]$ and $\mathsf{Cncrt}(P,t)$ from the same initial state are step-for-step identical. Thus, under $t$, the two programs have the same terminating and diverging executions.
\end{proof}

\begin{corollary}[Scenario-level loop-focused preservation under slicing and concretization]
\label{corollary:combination}
    Under Assumption~\ref{assumption:slicing-adequacy}, for every program $P$, loop $\ell \in \mathit{Loops}(P)$, input assignment $t \in \mathcal{T}$, and state $\sigma \in \Sigma(P,\ell)$:
    \begin{enumerate}[leftmargin=*,noitemsep,topsep=0pt]
        \item $P[t],\sigma \Uparrow_\ell \Longrightarrow \mathsf{Cncrt}(\mathsf{Slice}(P,\ell),t),\sigma \Uparrow_\ell$;
        \item $\mathsf{Cncrt}(\mathsf{Slice}(P,\ell),t),\sigma \Uparrow_\ell \Longrightarrow \exists \ell' \in \mathit{Loops}(P).\; \ell' \preceq^P_\sigma \ell \wedge \sigma \in \Sigma(P,\ell') \wedge P[t],\sigma \Uparrow_{\ell'}$;
        \item $\mathsf{Cncrt}(\mathsf{Slice}(P,\ell),t),\sigma \Downarrow_\ell \Longrightarrow P[t],\sigma \Downarrow_\ell$;
        \item $\bigl(\forall \ell' \in \mathit{Loops}(P).\; \ell' \preceq^P_\sigma \ell \wedge \sigma \in \Sigma(P,\ell') \Rightarrow P[t],\sigma \Downarrow_{\ell'}\bigr) \Longrightarrow \mathsf{Cncrt}(\mathsf{Slice}(P,\ell),t),\sigma \Downarrow_\ell$.
    \end{enumerate}
\end{corollary}

\begin{proof}[Proof sketch]
    The combined variant $\mathsf{Cncrt}(\mathsf{Slice}(P,\ell),t)$ is the loop-focused slice under the fixed input assignment $t$. By Theorem~\ref{theorem:concretization}, concretization preserves exactly the executions of the selected scenario. By Assumption~\ref{assumption:slicing-adequacy}, the slice preserves the criterion behavior needed for the selected loop obligation, with retained nondeterministic occurrences remaining compatible with the concretization mapping. Instantiating Theorems~\ref{theorem:slicing-nontermination} and~\ref{theorem:slicing-termination} within scenario $t$ yields the four claims. Thus, divergence of the selected original loop obligation under $t$ is preserved in the sliced-and-concretized variant, while divergence of the combined variant is reflected, under the same assignment, by divergence of $\ell$ itself or of a loop that can be encountered before $\ell$ from $\sigma$. The termination claims follow by the same contrapositive reasoning as in Theorem~\ref{theorem:slicing-termination}. Hence, the result is scenario-level and loop-focused, not a whole-program equivalence claim.
\end{proof}

\section{Implementation}
\label{section:implementation}
We implemented \tool as a C++23 source-level preprocessing framework for C programs, using Clang LibTooling for source-to-source transformations and Frama-C for slicing, without modifying downstream (non-)termination analyzers.

Loop-based slicing uses a Clang front end to identify \texttt{for}, \texttt{while}, and \texttt{do}--\texttt{while} loops, extract loop descriptors, and emit slicing metadata. It annotates loop-control statements (\texttt{continue}, \texttt{break}, \texttt{goto}, and \texttt{return}) to preserve termination-relevant control flow. The slicing backend invokes Frama-C in batch mode, running Evolved Value Analysis (EVA) with the slicing plugin to generate compilable sliced variants. For loops outside \texttt{main}, caller- and callee-context slices are merged through lightweight source transformations to retain relevant interprocedural context. ACSL annotations are removed after slicing for analyzers that do not accept ACSL.

Input-driven concretization is a Clang-based specialization pass. Given a tabular input specification, \tool replaces each \texttt{\_\_VERIFIER\_nondet\_*()} occurrence with the corresponding type-correct literal for a selected input assignment while preserving declarations and control structure. It produces one variant per input assignment, can be applied to original or sliced programs, and serves as input-scenario specialization rather than standalone input generation.

A top-level driver orchestrates preprocessing and analysis: for each program, it applies the selected configuration, generates variants, and runs the configured analyzers under fixed resource limits. It normalizes outcomes to \emph{yes}, \emph{no}, \emph{unknown}, \emph{timeout}, and \emph{error}, records wall-clock runtimes, generates machine-readable logs, and implements Algorithm~\ref{algorithm:tool-decision-evaluation} to aggregate outcomes against program- and variant-level ground truth, producing program-level decisions and $ratio_T$ and $ratio_{NT}$ values.

\section{Experimental Setup}
\label{section:experimental-setup}
This section describes the experimental setup, including benchmarks, tools, configurations, measurements, and research questions.

\subsection{Benchmarks}
\label{benchmarks}
We evaluate our approach on the benchmark suite introduced by Shi~et~al.~\cite{shi2022,fse}, which contains 117 simplified C/C++ programs derived from real-world OSS non-termination bugs: 56 non-terminating cases and 61 corresponding terminating fixes. The suite was constructed using explicit filtering and simplification criteria and includes program patterns that are challenging for many analyzers, such as pointers, arrays, structured data, bitwise operations, bounded arithmetic, and recursion. Accordingly, some tools in our study do not aim to provide sound (non-)termination guarantees for all benchmark patterns, consistent with observations by Shi~et~al.\footnote{Athena was designed to soundly support the semantic setting of the Shi~et~al. benchmark suite.} 
The benchmark provides program-level ground-truth labels\footnote{We use corrected labels for two mislabeled benchmarks: \texttt{Incorrect\_Initialization\_2\_T} (non-terminating) and \texttt{Signed\_Overflow\_Error\_1\_NT} (terminating).} for termination and non-termination.

Our preprocessing pipeline generates additional variants through slicing and input-driven concretization. For generated variants, we assign labels according to the semantics of the corresponding program instance, enabling consistent aggregation of verification outcomes across variants. For concretization, we generate ten input assignments per program. For non-terminating programs, these assignments cover both terminating and diverging executions when feasible, using five terminating and five non-terminating scenarios where such a split can be established; for terminating programs, all assignments correspond to terminating executions. This controlled construction exposes distinct execution domains while preserving comparability across programs.

\subsection{Tools}
\label{tools}
We evaluate six widely studied (non-)termination analyzers: Athena, Proton, UAutomizer, AProVE, CPAchecker, and 2LS. Together, they span diverse verification paradigms, including logic-based fixpoint reasoning, bounded-model-checking witness generation, trace abstraction with interpolation refinement, symbolic-execution-based termination analysis, configurable program analysis, and bounded model checking with k-induction and template-based invariant synthesis. Experiments were conducted on an Apple M2 Pro machine with a 10-core CPU and 16~GB RAM, running macOS~15.6.1 and 64-bit Linux Docker containers to provide a consistent execution environment. Tool configurations follow prior evaluations on this benchmark, and all tools use a uniform five-minute timeout per analyzed variant.

\subsection{Configurations}
\label{configurations}
For each program and analyzer, we evaluate four configurations that differ only in the preprocessing applied by \tool, with backend analyzers and parameters unchanged:

\begin{enumerate}[leftmargin=*,noitemsep,topsep=0pt]
    \item \base: analysis of the original program
    \item \slice: analysis of loop-based slices
    \item \concr: analysis of concretized variants
    \item \sliceconcr: analysis of variants obtained by applying slicing before concretization
\end{enumerate}

These configurations isolate the individual and combined effects of structural isolation and semantic specialization. 
The reverse order, \concrslice, may produce different slices because concretization can change reachability and dependence information. However, we evaluate \sliceconcr to avoid recomputing expensive loop-targeted slices for each input assignment.

\subsection{Measurements}
\label{outcomes-and-measurements}
For each analyzed program variant, we record the analyzer outcome---\emph{yes}, \emph{no}, \emph{unknown}, \emph{timeout}, or \emph{error}---and the wall-clock time, including preprocessing and backend analysis. Outcomes are aggregated following Algorithm~\ref{algorithm:tool-decision-evaluation}, yielding a program-level decision and the variant-resolution ratios $\mathit{ratio}_T$ and $\mathit{ratio}_{NT}$. For runtime, we report three per-program statistics: \emph{Total Variant Time (TVT)}, the cumulative runtime over generated variants; \emph{Average Variant Time (AVT)}, TVT divided by the number of variants; and \emph{Median Variant Time (MVT)}, which reduces sensitivity to outliers. To compare the cost of successful verification, we compute timing statistics over correctly solved programs. We also report the \emph{Solved Ratio (SR)}, the proportion of programs included in the timing analysis.

\subsection{Research Questions}
\label{research-questions}
Our evaluation is guided by the following research questions.

\begin{itemize}[leftmargin=*,noitemsep,topsep=0pt]
    \item \textbf{RQ1 (Accuracy Impact).}
    For each verification tool, to what extent do \slice, \concr, and \sliceconcr change the number of correctly classified (non-)terminating programs relative to \base, and what factors explain these changes?

    \item \textbf{RQ2 (Complementarity to \base).}
    To what extent do \slice, \concr, and \sliceconcr complement \base by correctly handling programs for which \base yields an incorrect or inconclusive outcome?
    
    \item \textbf{RQ3 (Loop-Level Localization of (Non-)Termination).}
    To what extent does \slice localize an inconclusive program-level outcome to specific loops by producing conclusive (non-)termination results for the remaining loops?

    \item \textbf{RQ4 (Feature Sensitivity).}
    Which structural and data-type characteristics are associated with successful (non-)termination analysis under \base, \slice, \concr, and \sliceconcr?
    
    \item \textbf{RQ5 (Incorrect Classification Reduction).}
    Do \slice, \concr, and \sliceconcr reduce incorrect (non-)termination classifications relative to \base?
    
    \item \textbf{RQ6 (Efficiency and Runtime Impact).}
    Compared to \base, how do \slice, \concr, and \sliceconcr affect TVT, AVT, and MVT for correctly solved programs, and how should these costs be interpreted with respect to SR?

    \item \textbf{RQ7 (Tractability--Generality Trade-offs).}
    What trade-offs arise between detectability, solver tractability, and the semantic scope over which correctness or divergence is established under \slice, \concr, and \sliceconcr?

    \item \textbf{RQ8 (Integration Potential).}
    Do the empirical effects of \slice, \concr, and \sliceconcr justify their integration into \base verification, and if so, how should they be incorporated?
\end{itemize}

\section{Results}
\label{section:results}
This section presents the empirical findings organized according to the research questions.

\begin{figure*}[!t]
\centering
\begin{subfigure}[t]{\textwidth}
\centering
\includegraphics[width=\textwidth]{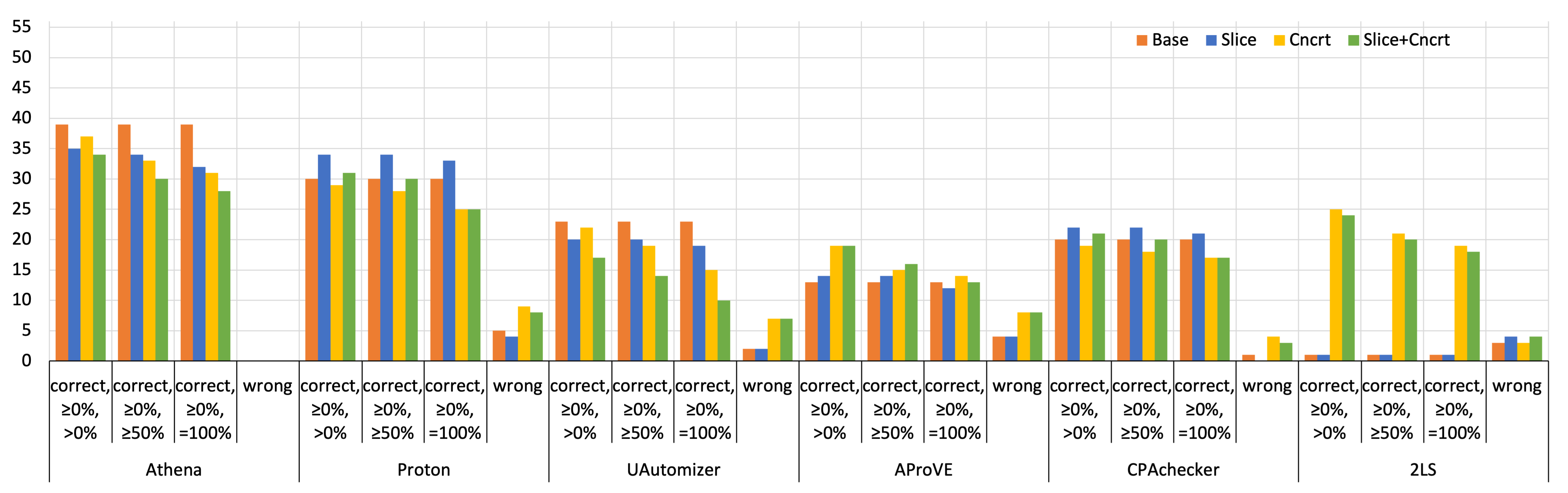}
\caption{Results for non-terminating benchmarks. $Correct,ratio_T\geq0\%,ratio_{NT}>0\%$ denotes program-level non-termination detection, since at least one non-terminating variant is classified as non-terminating. $Correct,ratio_T\geq0\%,ratio_{NT}\geq50\%$ and $Correct,ratio_T\geq0\%,ratio_{NT}=100\%$ denote the same program-level detection with stronger variant-level coverage, where at least half or all non-terminating variants are classified as non-terminating. $Wrong$ denotes a variant-level classification inconsistent with ground truth.}
\vspace{8pt}
\label{figure:rq1-nonterm}
\end{subfigure}
\begin{subfigure}[t]{\textwidth}
\centering
\includegraphics[width=\textwidth]{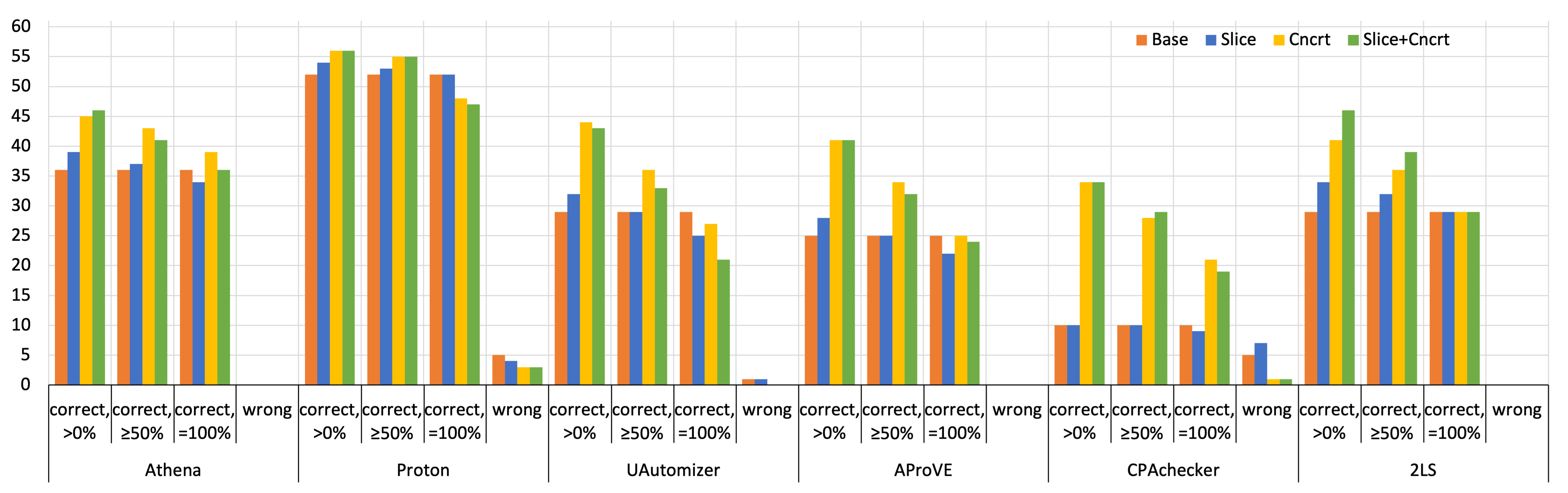}
\caption{Results for terminating benchmarks. $Correct,ratio_T=100\%$ denotes program-level termination resolution, since all generated variants are classified as terminating. $Correct,ratio_T\geq50\%$ and $Correct,ratio_T>0\%$ denote partial variant-level resolution, where at least half or at least one generated variant is classified as terminating, but program-level termination remains unresolved. $Wrong$ denotes a variant-level classification inconsistent with ground truth.}
\label{figure:rq1-term}
\end{subfigure}
\caption{Accuracy under preprocessing configurations. Bars report the number of benchmarks satisfying the corresponding outcome criterion for each analyzer and preprocessing configuration. Results for \concr and \sliceconcr are evaluated over selected input-scenario variants, not over all nondeterministic executions of the original program.}
\label{figure:rq1}
\end{figure*}

\subsection{RQ1: Accuracy Impact}
Figure~\ref{figure:rq1} compares outcome distributions across preprocessing configurations for each analyzer. For correct classifications, we focus on $Correct, ratio_T \geq 0\%, ratio_{NT} > 0\%$ for non-termination and $Correct, ratio_T = 100\%$ for termination; the remaining thresholds capture variant-level coverage.

The effect of preprocessing is strongly analyzer- and task-dependent: no configuration consistently improves accuracy across all analyzers or both verification tasks. For non-termination under \slice, Proton (\base 30 $\rightarrow$ \slice 34, +4), CPAchecker (+2), and AProVE (+1) improve, whereas Athena ($-4$) and UAutomizer ($-3$) decrease, and 2LS is unchanged. Under \concr, 2LS shows the largest gain (\base 1 $\rightarrow$ \concr 25, +24), followed by AProVE (+6), while the remaining tools decrease by $-1$ or $-2$. The task also matters: positive effects are more frequent for non-termination than for termination, with \slice improving three analyzers on non-termination---Proton (+4), CPAchecker (+2), and AProVE (+1)---but none on full termination resolution, and \sliceconcr improving four analyzers on non-termination---2LS (+23), AProVE (+6), Proton (+1), and CPAchecker (+1)---whereas only CPAchecker improves on termination (+9). These results suggest that preprocessing benefits depend on interactions among transformed programs, analyzer-specific abstractions and heuristics, and the reasoning demands of termination versus non-termination, so the same transformation may help some analyzers or tasks while having little effect or reducing accuracy for others.

Configurations containing \concr produce the largest improvements, but only for specific analyzer--task pairs. For non-termination, the largest gains are for 2LS, which increases from 1 under \base to 25 under \concr (+24) and 24 under \sliceconcr (+23), followed by AProVE, which improves from 13 to 19 under both configurations (+6). For termination, CPAchecker has the largest gains, increasing from 10 under \base to 21 under \concr (+11) and 19 under \sliceconcr (+9). Other tools show smaller gains, no change, or marginal decreases. These results suggest that \concr can help certain analyzers classify generated instances by reducing nondeterministic variability and specializing analysis to selected input-scenario variants.

In contrast, \slice produces smaller and more mixed changes. For non-termination, changes range from $-4$ to $+4$: Proton (\base 30 $\rightarrow$ \slice 34, +4), CPAchecker (20 $\rightarrow$ 22, +2), and AProVE (13 $\rightarrow$ 14, +1) improve, while Athena (39 $\rightarrow$ 35, $-4$) and UAutomizer (23 $\rightarrow$ 20, $-3$) decrease. For termination, \slice does not increase the number of fully resolved programs for any analyzer; the changes are either small decreases for UAutomizer ($-4$), AProVE ($-3$), Athena ($-2$), and CPAchecker ($-1$), or no change for Proton and 2LS. Compared with \concr-based configurations, \slice has a smaller effect, suggesting that loop-focused structural isolation can help some analyzers but remains limited because it preserves broader loop-relevant behavior and nondeterminism.

Finally, the combination \sliceconcr is not consistently additive. Compared with \slice, it improves 2 of 6 analyzers on non-termination---2LS (\slice 1 $\rightarrow$ \sliceconcr 24, +23) and AProVE (14 $\rightarrow$ 19, +5)---and 3 of 6 on termination---CPAchecker (+10), Athena (+2), and AProVE (+2). Compared with \concr, however, it improves only Proton and CPAchecker on non-termination (+2 each), improves none on termination, and decreases several analyzer--task pairs. Thus, structural isolation and input-scenario specialization are not necessarily additive: depending on analyzer abstractions, heuristics, and reasoning mechanisms, their combination may help or reduce the number of correctly classified programs.

\begin{rqtakeaway}
\textbf{RQ1 Takeaway.}
Preprocessing can improve accuracy, but gains are concentrated in specific analyzer--task pairs rather than being uniform. The largest improvements come from configurations containing \concr, while \slice produces more modest changes and \sliceconcr is not consistently additive.
\end{rqtakeaway}

\begin{table*}[!t]
\centering
\footnotesize
\setlength{\tabcolsep}{7pt}
\caption{Complementarity to the \base\ configuration. Rows denote the classification produced by \base, and columns denote the classification produced by the applied preprocessing configuration. Each entry reports the number of benchmarks in the corresponding row--column combination for a given analyzer. Transitions from U or W to S indicate recoveries, i.e., benchmarks unresolved or incorrectly classified by \base but solved after preprocessing, and quantify the complementarity provided by preprocessing. Transitions from S to U or W indicate losses, i.e., benchmarks solved by \base that become unresolved or incorrectly classified after preprocessing. Results for \concr\ and \sliceconcr\ are evaluated over selected input-scenario variants, not over all nondeterministic executions of the original program.}
\label{table:rq2}
\begin{subtable}[t]{\textwidth}
\centering
\caption{Results for non-terminating benchmarks. S (Solved) denotes program-level non-termination detection, i.e., at least one non-terminating variant is classified as non-terminating ($Correct,ratio_T\geq0\%,ratio_{NT}>0\%$). U (Unknown) denotes program-level unresolved cases, i.e., no non-terminating variant is classified as non-terminating ($Correct,ratio_T\geq0\%,ratio_{NT}=0\%$). W (Wrong) denotes a program-level classification inconsistent with the benchmark ground truth.}
\vspace{-4pt}
\label{table:rq2-nonterm}
\begin{tabular}{l|c|ccc|ccc|ccc|ccc|ccc|ccc}
\noalign{\hrule height 1pt}
\multicolumn{1}{c|}{\multirow{2}{*}{\shortstack[c]{\textbf{Applied}\\\textbf{Configuration}}}} &
\multicolumn{1}{c|}{\multirow{2}{*}{\shortstack[c]{\textbf{\base}\\\textbf{Classification}}}} &
\multicolumn{3}{c|}{\textbf{Athena}} &
\multicolumn{3}{c|}{\textbf{Proton}} &
\multicolumn{3}{c|}{\textbf{UAutomizer}} &
\multicolumn{3}{c|}{\textbf{AProVE}} &
\multicolumn{3}{c|}{\textbf{CPAchecker}} &
\multicolumn{3}{c}{\textbf{2LS}} \\
\cline{3-20}
& & S & U & W & S & U & W & S & U & W & S & U & W & S & U & W & S & U & W \\
\noalign{\hrule height 1pt}
\multirow{3}{*}{\slice}
& S & 35 & 4 & 0 & 29 & 1 & 0 & 18 & 5 & 0 & 12 & 1 & 0 & 18 & 2 & 0 & 1 & 0 & 0 \\
& U & 0 & 17 & 0 & 4 & 16 & 1 & 2 & 29 & 0 & 2 & 36 & 1 & 3 & 32 & 0 & 0 & 51 & 1 \\
& W & 0 & 0 & 0 & 1 & 1 & 3 & 0 & 0 & 2 & 0 & 1 & 3 & 1 & 0 & 0 & 0 & 0 & 3 \\
\noalign{\hrule height 0.6pt}
\multirow{3}{*}{\concr}
& S & 36 & 3 & 0 & 27 & 0 & 3 & 20 & 1 & 2 & 11 & 0 & 2 & 18 & 0 & 2 & 1 & 0 & 0 \\
& U & 1 & 16 & 0 & 2 & 18 & 1 & 2 & 26 & 3 & 8 & 27 & 4 & 1 & 33 & 1 & 24 & 28 & 0 \\
& W & 0 & 0 & 0 & 0 & 0 & 5 & 0 & 0 & 2 & 0 & 2 & 2 & 0 & 0 & 1 & 0 & 0 & 3 \\
\noalign{\hrule height 0.6pt}
\multirow{3}{*}{\sliceconcr}
& S & 33 & 6 & 0 & 24 & 1 & 5 & 14 & 7 & 2 & 10 & 0 & 3 & 17 & 1 & 2 & 1 & 0 & 0 \\
& U & 1 & 16 & 0 & 6 & 13 & 2 & 3 & 25 & 3 & 9 & 26 & 4 & 3 & 31 & 1 & 23 & 28 & 1 \\
& W & 0 & 0 & 0 & 1 & 3 & 1 & 0 & 0 & 2 & 0 & 3 & 1 & 1 & 0 & 0 & 0 & 0 & 3 \\
\noalign{\hrule height 1pt}
\end{tabular}
\end{subtable}
\begin{subtable}[t]{\textwidth}
\centering
\vspace{6pt}
\caption{Results for terminating benchmarks. S (Solved) denotes program-level termination resolution, i.e., all generated variants are classified as terminating ($Correct,ratio_T=100\%$). U (Unknown) denotes program-level unresolved cases, i.e., fewer than all generated variants are classified as terminating ($Correct,ratio_T<100\%$). W (Wrong) denotes a program-level classification inconsistent with the benchmark ground truth.}
\vspace{-4pt}
\label{table:rq2-term}
\begin{tabular}{l|c|ccc|ccc|ccc|ccc|ccc|ccc}
\noalign{\hrule height 1pt}
\multicolumn{1}{c|}{\multirow{2}{*}{\shortstack[c]{\textbf{Applied}\\\textbf{Configuration}}}} &
\multicolumn{1}{c|}{\multirow{2}{*}{\shortstack[c]{\textbf{\base}\\\textbf{Classification}}}} &
\multicolumn{3}{c|}{\textbf{Athena}} &
\multicolumn{3}{c|}{\textbf{Proton}} &
\multicolumn{3}{c|}{\textbf{UAutomizer}} &
\multicolumn{3}{c|}{\textbf{AProVE}} &
\multicolumn{3}{c|}{\textbf{CPAchecker}} &
\multicolumn{3}{c}{\textbf{2LS}} \\
\cline{3-20}
& & S & U & W & S & U & W & S & U & W & S & U & W & S & U & W & S & U & W \\
\noalign{\hrule height 1pt}
\multirow{3}{*}{\slice}
& S & 34 & 2 & 0 & 52 & 0 & 0 & 25 & 4 & 0 & 22 & 3 & 0 & 8 & 2 & 0 & 29 & 0 & 0 \\
& U & 0 & 25 & 0 & 0 & 4 & 0 & 0 & 31 & 0 & 0 & 36 & 0 & 1 & 43 & 2 & 0 & 32 & 0 \\
& W & 0 & 0 & 0 & 0 & 1 & 4 & 0 & 0 & 1 & 0 & 0 & 0 & 0 & 0 & 5 & 0 & 0 & 0 \\
\noalign{\hrule height 0.6pt}
\multirow{3}{*}{\concr}
& S & 35 & 1 & 0 & 45 & 7 & 0 & 25 & 4 & 0 & 23 & 2 & 0 & 8 & 2 & 0 & 28 & 1 & 0 \\
& U & 4 & 21 & 0 & 1 & 3 & 0 & 2 & 29 & 0 & 2 & 34 & 0 & 9 & 37 & 0 & 1 & 31 & 0 \\
& W & 0 & 0 & 0 & 2 & 0 & 3 & 0 & 1 & 0 & 0 & 0 & 0 & 4 & 0 & 1 & 0 & 0 & 0 \\
\noalign{\hrule height 0.6pt}
\multirow{3}{*}{\sliceconcr}
& S & 33 & 3 & 0 & 45 & 7 & 0 & 20 & 9 & 0 & 22 & 3 & 0 & 6 & 4 & 0 & 28 & 1 & 0 \\
& U & 3 & 22 & 0 & 1 & 3 & 0 & 1 & 30 & 0 & 2 & 34 & 0 & 9 & 37 & 0 & 1 & 31 & 0 \\
& W & 0 & 0 & 0 & 1 & 1 & 3 & 0 & 1 & 0 & 0 & 0 & 0 & 4 & 0 & 1 & 0 & 0 & 0 \\
\noalign{\hrule height 1pt}
\end{tabular}
\end{subtable}
\end{table*}

\subsection{RQ2: Complementarity to the \base}
Table~\ref{table:rq2} compares preprocessing classifications with the corresponding \base classifications for each analyzer. The results show that preprocessing complements \base by solving different benchmark subsets, not simply by adding cases already solved by \base. For example, on non-termination, \slice recovers 5 benchmarks for Proton---4 previously \emph{unknown} and 1 previously \emph{wrong}---while 1 previously \emph{solved} benchmark becomes \emph{unknown}. In contrast, for UAutomizer, \slice recovers 2 benchmarks but loses 5 previously solved ones. Similarly, on termination, CPAchecker has 13 recoveries under both \concr and \sliceconcr, but losses differ: 2 under \concr and 4 under \sliceconcr. Thus, preprocessing can both recover and lose cases, changing the benchmark subset handled by each analyzer. This occurs because some transformed cases align better with an analyzer's abstractions, heuristics, or reasoning mechanisms than the \base form, while others align less well.

Complementarity also differs by benchmark class and recovery source. For non-termination, the largest recovery rates are U$\rightarrow$S under \sliceconcr and \concr: 23.4\% and 19.8\%, respectively. For \sliceconcr, 23.4\% corresponds to 45 recoveries out of 192 \base-\emph{unknown} non-termination cases: Athena 1, Proton 6, UAutomizer 3, AProVE 9, CPAchecker 3, and 2LS 23. In contrast, the largest W$\rightarrow$S recovery rate for non-termination is 13.3\%, under both \slice and \sliceconcr. For termination, the pattern reverses: \concr and \sliceconcr recover 54.5\% and 45.5\% of \base-\emph{wrong} cases, compared with 11.0\% and 9.8\% of \base-\emph{unknown} cases. Thus, preprocessing mainly resolves previously inconclusive \base outcomes for non-termination, but more strongly recovers previously incorrect outcomes for termination.

\concr provides the strongest complementarity, concentrated in a few analyzer--task pairs, while \slice yields smaller and more mixed effects. Under \concr, the largest recoveries occur for 2LS on non-termination (24: 23 previously \emph{unknown} and 1 previously \emph{wrong}), CPAchecker on termination (13), and AProVE on non-termination (8); other recoveries range from 1 to 4. In contrast, \slice mainly recovers non-termination cases: Proton recovers 5, CPAchecker 4, and UAutomizer and AProVE 2 each, while termination recovery is limited to 1 case for CPAchecker. This reflects the transformations: \concr specializes verification to selected input-scenario variants, whereas \slice simplifies program structure while preserving broader behavior and nondeterminism.

Finally, \sliceconcr improves substantially over \slice, but not consistently over \concr. Compared with \slice, the largest recovery increases are 2LS on non-termination (\slice 0 $\rightarrow$ \sliceconcr 23, +23), CPAchecker on termination (+12), and AProVE on non-termination (+7). Compared with \concr, gains are smaller: on non-termination, Proton (\concr 2 $\rightarrow$ \sliceconcr 7, +5), CPAchecker (+3), UAutomizer (+1), and AProVE (+1) improve, while 2LS decreases slightly ($-1$) and Athena is unchanged. On termination, \sliceconcr provides no additional recoveries over \concr; results are unchanged or decrease by 1. These results suggest that most complementarity from \sliceconcr comes from concretization: adding slicing helps relative to \slice alone, but gives little additional benefit once \concr has already been applied.

\begin{rqtakeaway}
\textbf{RQ2 Takeaway.}
Preprocessing complements \base selectively: it solves different benchmark subsets and recovers different types of \base failures across benchmark classes. The strongest recoveries come from configurations containing \concr, while \slice is more limited and \sliceconcr mainly improves over \slice, not consistently over \concr.
\end{rqtakeaway}

\begin{table*}[!t]
\centering
\footnotesize
\setlength{\tabcolsep}{4.6pt}
\caption{Loop-level localization under slicing. The table groups benchmarks by the number of loops $n$ and reports, for each analyzer, how many benchmarks have exactly $k$ resolved loop-specific slices. Resolving $k$ loop-specific slices localizes the original program-level unknown outcome to the remaining $n-k$ unresolved loop-specific slices. In this benchmark set, the maximum number of loops is four. We define partial localization for multi-loop benchmarks as $1 < n \leq 4$ and $0 < k < n$, and complete localization as $1 \leq n \leq 4$ and $k = n$.}
\label{table:rq3}
\begin{subtable}[t]{\textwidth}
\centering
\caption{Results for non-terminating benchmarks.}
\vspace{-4pt}
\label{table:rq3-nonterm}
\begin{tabular}{c|c|ccccc|ccccc|ccccc|ccccc|ccccc|ccccc}
\noalign{\hrule height 1pt}
\multirow{3}{*}{\textbf{Loops}} & \multirow{3}{*}{\textbf{Programs}}
& \multicolumn{5}{c|}{\textbf{Athena}}
& \multicolumn{5}{c|}{\textbf{Proton}}
& \multicolumn{5}{c|}{\textbf{UAutomizer}}
& \multicolumn{5}{c|}{\textbf{AProVE}}
& \multicolumn{5}{c|}{\textbf{CPAchecker}}
& \multicolumn{5}{c}{\textbf{2LS}} \\
\cline{3-32}
& & \multicolumn{5}{c|}{Resolved Loops}
  & \multicolumn{5}{c|}{Resolved Loops}
  & \multicolumn{5}{c|}{Resolved Loops}
  & \multicolumn{5}{c|}{Resolved Loops}
  & \multicolumn{5}{c|}{Resolved Loops}
  & \multicolumn{5}{c}{Resolved Loops} \\
& & 0&1&2&3&4 & 0&1&2&3&4 & 0&1&2&3&4
  & 0&1&2&3&4 & 0&1&2&3&4 & 0&1&2&3&4 \\
\noalign{\hrule height 1pt}
1 & 37 & 14 & 23 & -- & -- & -- & 13 & 20 & -- & -- & -- & 25 & 10 & -- & -- & -- & 25 &  8 & -- & -- & -- & 21 & 16 & -- & -- & -- & 34 &  0 & -- & -- & -- \\
2 & 13 &  1 &  5 &  7 & -- & -- &  0 &  1 & 12 & -- & -- &  2 &  2 &  9 & -- & -- &  5 &  4 &  4 & -- & -- &  8 &  0 &  5 & -- & -- & 10 &  3 &  0 & -- & -- \\
3 & 4 &  1 &  1 &  0 &  2 & -- &  0 &  2 &  1 &  1 & -- &  0 &  3 &  1 &  0 & -- &  0 &  3 &  1 &  0 & -- &  3 &  0 &  1 &  0 & -- &  1 &  1 &  0 &  1 & -- \\
4 & 2 &  1 &  0 &  1 &  0 &  0 &  1 &  1 &  0 &  0 &  0 &  1 &  1 &  0 &  0 &  0 &  1 &  1 &  0 &  0 &  0 &  2 &  0 &  0 &  0 &  0 &  1 &  1 &  0 &  0 &  0 \\
\noalign{\hrule height 1pt}
\end{tabular}
\end{subtable}
\begin{subtable}[t]{\textwidth}
\centering
\vspace{6pt}
\caption{Results for terminating benchmarks.}
\vspace{-4pt}
\label{table:rq3-term}
\begin{tabular}{c|c|ccccc|ccccc|ccccc|ccccc|ccccc|ccccc}
\noalign{\hrule height 1pt}
\multirow{3}{*}{\textbf{Loops}} & \multirow{3}{*}{\textbf{Programs}}
& \multicolumn{5}{c|}{\textbf{Athena}}
& \multicolumn{5}{c|}{\textbf{Proton}}
& \multicolumn{5}{c|}{\textbf{UAutomizer}}
& \multicolumn{5}{c|}{\textbf{AProVE}}
& \multicolumn{5}{c|}{\textbf{CPAchecker}}
& \multicolumn{5}{c}{\textbf{2LS}} \\
\cline{3-32}
& & \multicolumn{5}{c|}{Resolved Loops}
  & \multicolumn{5}{c|}{Resolved Loops}
  & \multicolumn{5}{c|}{Resolved Loops}
  & \multicolumn{5}{c|}{Resolved Loops}
  & \multicolumn{5}{c|}{Resolved Loops}
  & \multicolumn{5}{c}{Resolved Loops} \\
& & 0&1&2&3&4 & 0&1&2&3&4 & 0&1&2&3&4
  & 0&1&2&3&4 & 0&1&2&3&4 & 0&1&2&3&4 \\
\noalign{\hrule height 1pt}
1 & 41 & 16 & 25 & -- & -- & -- &  3 & 36 & -- & -- & -- & 26 & 14 & -- & -- & -- & 27 & 14 & -- & -- & -- & 28 &  6 & -- & -- & -- & 20 & 21 & -- & -- & -- \\
2 & 14 &  5 &  3 &  6 & -- & -- &  0 &  0 & 12 & -- & -- &  2 &  3 &  9 & -- & -- &  5 &  2 &  7 & -- & -- & 11 &  0 &  3 & -- & -- &  6 &  2 &  6 & -- & -- \\
3 & 4 &  1 &  1 &  0 &  2 & -- &  0 &  0 &  1 &  3 & -- &  0 &  2 &  1 &  1 & -- &  0 &  2 &  1 &  1 & -- &  3 &  0 &  1 &  0 & -- &  1 &  1 &  1 &  1 & -- \\
4 & 2 &  0 &  1 &  0 &  0 &  1 &  0 &  1 &  0 &  0 &  1 &  0 &  1 &  0 &  0 &  1 &  1 &  1 &  0 &  0 &  0 &  2 &  0 &  0 &  0 &  0 &  0 &  1 &  0 &  0 &  1 \\
\noalign{\hrule height 1pt}
\end{tabular}
\end{subtable}
\end{table*}

\subsection{RQ3: Loop-Level Localization of (Non-)Termination}
Table~\ref{table:rq3} summarizes loop-level localization under \slice for benchmarks with different numbers of loops. Although \slice has limited program-level impact in RQ1 and RQ2, it provides partial loop-level localization for many multi-loop benchmarks. For non-termination, partial localization is most frequent for AProVE (9 benchmarks), followed by Athena and UAutomizer (7 each). For termination, it is most frequent for UAutomizer (7), followed by AProVE (6) and Athena and 2LS (5 each). For example, Athena's 7 non-termination cases comprise 5 two-loop benchmarks with one resolved slice, 1 three-loop benchmark with one resolved slice, and 1 four-loop benchmark with two resolved slices, localizing the remaining uncertainty to 1, 2, and 2 unresolved loops, respectively. Thus, \slice can provide diagnostic value even without full program-level resolution, because resolved slices reduce the set of loops that may be responsible for an \emph{unknown} outcome.

Complete localization becomes less frequent as loop count increases, whereas partial localization remains common in multi-loop benchmarks. For non-termination, complete localization reaches 34.7\% for 1-loop and 47.4\% for 2-loop benchmarks, but drops to 16.7\% for 3-loop and 0\% for 4-loop benchmarks. In contrast, partial localization remains visible for larger multi-loop benchmarks: 19.2\%, 58.3\%, and 41.7\% for 2-, 3-, and 4-loop benchmarks, respectively. The 47.4\% complete-localization rate for 2-loop benchmarks corresponds to 37 of $13 \times 6$ analyzer--benchmark pairs in which the whole program, including both loops, is resolved: 7 Athena, 12 Proton, 9 UAutomizer, 4 AProVE, 5 CPAchecker, and 0 2LS. The 58.3\% partial-localization rate for 3-loop benchmarks corresponds to 14 of $4 \times 6$ pairs in which at least one, but not all, loop-specific slices are resolved: 1 Athena, 3 Proton, 4 UAutomizer, 4 AProVE, 1 CPAchecker, and 1 2LS. This difference occurs because complete localization requires resolving all loops, which becomes harder as loop count increases, whereas partial localization only requires resolving a subset of loop-specific slices, each analyzed separately using its loop-relevant context.

Finally, partial localization is analyzer- and task-dependent. AProVE has the largest partial-localization count for non-termination (9), UAutomizer has the largest count for termination (7), and CPAchecker has the smallest count in both benchmark classes (1). Moreover, every analyzer has equal or higher partial-localization counts for non-termination than for termination: Athena (7 vs. 5), Proton (5 vs. 2), UAutomizer (7 vs. 7), AProVE (9 vs. 6), CPAchecker (1 vs. 1), and 2LS (5 vs. 5). Thus, the same loop-level decomposition provides different localization benefits across analyzers, and \slice narrows unresolved non-termination cases at least as often as termination cases.

\begin{rqtakeaway}
\textbf{RQ3 Takeaway.}
\slice provides diagnostic value through loop-level localization rather than program-level resolution. Complete localization decreases as loop count increases, but partial localization remains common. Both effects are analyzer-dependent and stronger for non-termination than termination.
\end{rqtakeaway}

\begin{table*}[!t]
\centering
\footnotesize
\setlength{\tabcolsep}{7pt}
\caption{Structural sensitivity under preprocessing configurations. Rows group benchmarks by structural or data-type feature, and entries report the number of benchmarks classified as S, U, or W by each analyzer under each configuration. S, U, and W follow the definitions in Table~\ref{table:rq2}. Results for \concr and \sliceconcr are evaluated over selected input-scenario variants, not over all nondeterministic executions of the original program.}
\label{table:rq4}
\begin{subtable}[t]{\textwidth}
\centering
\caption{Results for non-terminating benchmarks.}
\vspace{-4pt}
\label{table:rq4-nonterm}
\begin{tabular}{l|l|ccc|ccc|ccc|ccc|ccc|ccc}
\noalign{\hrule height 1pt}
\multicolumn{1}{c|}{\multirow{2}{*}{\shortstack[c]{\textbf{Features (Total)}}}} &
\multicolumn{1}{c|}{\multirow{2}{*}{\shortstack[c]{\textbf{Configuration}}}} &
\multicolumn{3}{c|}{\textbf{Athena}} &
\multicolumn{3}{c|}{\textbf{Proton}} &
\multicolumn{3}{c|}{\textbf{UAutomizer}} &
\multicolumn{3}{c|}{\textbf{AProVE}} &
\multicolumn{3}{c|}{\textbf{CPAchecker}} &
\multicolumn{3}{c}{\textbf{2LS}} \\
\cline{3-20}
& & S & U & W & S & U & W & S & U & W & S & U & W & S & U & W & S & U & W \\
\noalign{\hrule height 1pt}
\multirow{4}{*}{\shortstack[l]{Pointer\\Manipulation (7)}} & \base        & 4&3&0 & 4&3&0 & 2&5&0 & 0&7&0 & 0&7&0 & 0&7&0 \\
                                                           & \slice       & 2&5&0 & 5&2&0 & 2&5&0 & 0&7&0 & 0&7&0 & 0&6&1 \\
                                                           & \concr       & 2&5&0 & 6&1&0 & 2&3&2 & 2&5&0 & 0&7&0 & 1&6&0 \\
                                                           & \sliceconcr & 1&6&0 & 6&0&1 & 2&3&2 & 2&5&0 & 0&7&0 & 1&5&1 \\
\noalign{\hrule height 0.6pt}
\multirow{4}{*}{\shortstack[l]{Array (8)}}                 & \base        & 4&4&0 & 3&5&0 & 2&6&0 & 0&8&0 & 0&8&0 & 0&8&0 \\
                                                           & \slice       & 3&5&0 & 4&4&0 & 2&6&0 & 1&7&0 & 0&8&0 & 0&7&1 \\
                                                           & \concr       & 3&5&0 & 5&2&1 & 2&3&3 & 1&6&1 & 0&8&0 & 2&6&0 \\
                                                           & \sliceconcr & 3&5&0 & 4&1&3 & 1&3&4 & 1&5&2 & 0&8&0 & 1&6&1 \\
\noalign{\hrule height 0.6pt}
\multirow{4}{*}{\shortstack[l]{Data\\Structure (3)}}       & \base        & 2&1&0 & 3&0&0 & 2&1&0 & 0&3&0 & 0&3&0 & 0&3&0 \\
                                                           & \slice       & 1&2&0 & 3&0&0 & 2&1&0 & 0&3&0 & 0&3&0 & 0&3&0 \\
                                                           & \concr       & 2&1&0 & 3&0&0 & 2&1&0 & 1&2&0 & 0&3&0 & 1&2&0 \\
                                                           & \sliceconcr & 1&2&0 & 3&0&0 & 2&1&0 & 1&2&0 & 0&3&0 & 1&2&0 \\
\noalign{\hrule height 0.6pt}
\multirow{4}{*}{\shortstack[l]{Bit\\Calculation (5)}}      & \base        & 5&0&0 & 4&1&0 & 2&2&1 & 0&5&0 & 5&0&0 & 0&5&0 \\
                                                           & \slice       & 5&0&0 & 5&0&0 & 2&2&1 & 0&5&0 & 5&0&0 & 0&5&0 \\
                                                           & \concr       & 5&0&0 & 4&1&0 & 2&2&1 & 0&5&0 & 5&0&0 & 5&0&0 \\
                                                           & \sliceconcr & 5&0&0 & 5&0&0 & 2&2&1 & 0&5&0 & 5&0&0 & 5&0&0 \\
\noalign{\hrule height 0.6pt}
\multirow{4}{*}{\shortstack[l]{Integer\\Underflow/\\Overflow (14)}} &
\base        & 14&0&0 & 9&1&4 & 3&9&2 & 1&9&4 & 7&6&1 & 1&13&0 \\
& \slice       & 13&1&0 & 9&3&2 & 4&8&2 & 2&8&4 & 10&4&0 & 1&13&0 \\
& \concr       & 14&0&0 & 9&1&4 & 2&10&2 & 2&9&3 & 7&5&2 & 10&4&0 \\
& \sliceconcr & 13&1&0 & 9&5&0 & 3&9&2 & 3&10&1 & 9&4&1 & 10&4&0 \\
\noalign{\hrule height 0.6pt}
\multirow{4}{*}{\shortstack[l]{Recursion (12)}}            & \base        & 0&12&0 & 0&11&1 & 4&8&0 & 2&10&0 & 0&12&0 & 0&9&3 \\
                                                           & \slice       & 0&12&0 & 0&11&1 & 0&12&0 & 2&10&0 & 0&12&0 & 0&9&3 \\
                                                           & \concr       & 0&12&0 & 0&11&1 & 4&8&0 & 4&8&0 & 0&12&0 & 0&9&3 \\
                                                           & \sliceconcr & 0&12&0 & 0&11&1 & 0&12&0 & 4&8&0 & 0&12&0 & 0&9&3 \\
\noalign{\hrule height 1pt}
\end{tabular}
\end{subtable}
\begin{subtable}[t]{\textwidth}
\centering
\vspace{6pt}
\caption{Results for terminating benchmarks.}
\vspace{-4pt}
\label{table:rq4-term}
\begin{tabular}{l|l|ccc|ccc|ccc|ccc|ccc|ccc}
\noalign{\hrule height 1pt}
\multicolumn{1}{c|}{\multirow{2}{*}{\shortstack[c]{\textbf{Features (Total)}}}} &
\multicolumn{1}{c|}{\multirow{2}{*}{\shortstack[c]{\textbf{Configuration}}}} &
\multicolumn{3}{c|}{\textbf{Athena}} &
\multicolumn{3}{c|}{\textbf{Proton}} &
\multicolumn{3}{c|}{\textbf{UAutomizer}} &
\multicolumn{3}{c|}{\textbf{AProVE}} &
\multicolumn{3}{c|}{\textbf{CPAchecker}} &
\multicolumn{3}{c}{\textbf{2LS}} \\
\cline{3-20}
& & S & U & W & S & U & W & S & U & W & S & U & W & S & U & W & S & U & W \\
\noalign{\hrule height 1pt}
\multirow{4}{*}{\shortstack[l]{Pointer\\Manipulation (7)}} & \base        & 0&7&0 & 2&2&3 & 1&6&0 & 1&6&0 & 0&7&0 & 0&7&0 \\
                                                           & \slice       & 0&7&0 & 2&2&3 & 1&6&0 & 1&6&0 & 0&7&0 & 0&7&0 \\
                                                           & \concr       & 1&6&0 & 4&2&1 & 2&5&0 & 1&6&0 & 0&7&0 & 0&7&0 \\
                                                           & \sliceconcr & 1&6&0 & 3&2&2 & 2&5&0 & 1&6&0 & 0&7&0 & 0&7&0 \\
\noalign{\hrule height 0.6pt}
\multirow{4}{*}{\shortstack[l]{Array (8)}}                 & \base        & 4&4&0 & 5&2&1 & 5&3&0 & 3&5&0 & 0&8&0 & 3&5&0 \\
                                                           & \slice       & 3&5&0 & 5&2&1 & 5&3&0 & 2&6&0 & 0&8&0 & 3&5&0 \\
                                                           & \concr       & 4&4&0 & 7&1&0 & 4&4&0 & 3&5&0 & 0&8&0 & 3&5&0 \\
                                                           & \sliceconcr & 3&5&0 & 6&1&1 & 4&4&0 & 3&5&0 & 0&8&0 & 3&5&0 \\
\noalign{\hrule height 0.6pt}
\multirow{4}{*}{\shortstack[l]{Data\\Structure (3)}}       & \base        & 0&3&0 & 1&0&2 & 0&3&0 & 0&3&0 & 0&3&0 & 0&3&0 \\
                                                           & \slice       & 0&3&0 & 1&0&2 & 0&3&0 & 0&3&0 & 0&3&0 & 0&3&0 \\
                                                           & \concr       & 1&2&0 & 1&1&1 & 0&3&0 & 0&3&0 & 0&3&0 & 0&3&0 \\
                                                           & \sliceconcr & 1&2&0 & 1&1&1 & 0&3&0 & 0&3&0 & 0&3&0 & 0&3&0 \\
\noalign{\hrule height 0.6pt}
\multirow{4}{*}{\shortstack[l]{Bit\\Calculation (13)}}     & \base        & 10&3&0 & 13&0&0 & 6&6&1 & 1&12&0 & 1&9&3 & 6&7&0 \\
                                                           & \slice       & 10&3&0 & 13&0&0 & 6&6&1 & 1&12&0 & 1&8&4 & 6&7&0 \\
                                                           & \concr       & 10&3&0 & 13&0&0 & 5&7&0 & 1&12&0 & 9&4&0 & 6&7&0 \\
                                                           & \sliceconcr & 10&3&0 & 13&0&0 & 5&8&0 & 1&12&0 & 9&4&0 & 6&7&0 \\
\noalign{\hrule height 0.6pt}
\multirow{4}{*}{\shortstack[l]{Integer\\Underflow/\\Overflow (21)}} &
\base        & 17&4&0 & 20&0&1 & 8&12&1 & 7&14&0 & 4&14&3 & 13&8&0 \\
& \slice       & 16&5&0 & 20&1&0 & 8&12&1 & 6&15&0 & 4&12&5 & 13&8&0 \\
& \concr       & 17&4&0 & 17&3&1 & 7&13&0 & 6&15&0 & 10&10&1 & 12&9&0 \\
& \sliceconcr & 15&6&0 & 17&4&0 & 7&14&0 & 5&16&0 & 10&10&1 & 12&9&0 \\
\noalign{\hrule height 0.6pt}
\multirow{4}{*}{\shortstack[l]{Recursion (12)}}            & \base        & 0&12&0 & 10&2&0 & 3&9&0 & 3&9&0 & 0&12&0 & 4&8&0 \\
                                                           & \slice       & 0&12&0 & 10&2&0 & 0&12&0 & 3&9&0 & 0&12&0 & 4&8&0 \\
                                                           & \concr       & 0&12&0 & 9&3&0 & 3&9&0 & 3&9&0 & 0&12&0 & 4&8&0 \\
                                                           & \sliceconcr & 0&12&0 & 9&3&0 & 0&12&0 & 3&9&0 & 0&12&0 & 4&8&0 \\
\noalign{\hrule height 1pt}
\end{tabular}
\end{subtable}
\end{table*}

\subsection{RQ4: Feature Sensitivity}
Table~\ref{table:rq4} reports outcomes for benchmark groups characterized by different structural and data-type features under each preprocessing configuration. The results show that preprocessing is feature-dependent, and that responsive feature categories vary across analyzers and verification tasks. For example, on non-termination, AProVE improves on Pointer Manipulation from 0 solved programs under \base to 2 under both \concr and \sliceconcr, but shows no improvement on Bit Calculation, where all configurations leave 5 programs \emph{unknown}. Thus, for the same analyzer, preprocessing helps some feature categories but not others. Responsive features also differ across analyzers: although AProVE does not improve on Bit Calculation, 2LS improves from 0 solved programs under \base to 5 under both \concr and \sliceconcr. The task also matters: although AProVE improves on non-termination Pointer Manipulation, the corresponding termination benchmarks show no improvement, with all configurations solving 1 program and leaving 6 \emph{unknown}. This suggests that preprocessing helps when transformed tasks match analyzer-specific feature reasoning, but may have little effect on poorly handled features; termination and non-termination reasoning can further change which features benefit.

Preprocessing is most effective on Integer Underflow/Overflow and Bit Calculation benchmarks, while Recursion and Data Structure benchmarks respond weakest. On non-termination, Integer Underflow/Overflow shows gains for four analyzers: UAutomizer (+1 under \slice; \base 3 $\rightarrow$ \slice 4), AProVE (+1, +1, and +2 under \slice, \concr, and \sliceconcr), CPAchecker (+3 and +2 under \slice and \sliceconcr), and 2LS (+9 under both \concr and \sliceconcr). In contrast, Recursion improves only for AProVE (+2 under both \concr and \sliceconcr), with no gains for Athena, Proton, CPAchecker, or 2LS. These results suggest that preprocessing is most beneficial when the main challenge involves arithmetic or bit-level behavior that can be simplified or specialized, whereas recursion and complex data structures require reasoning capabilities that preprocessing alone does not substantially provide.

Finally, configurations containing \concr produce the largest feature-specific improvements, while \slice alone yields smaller and more mixed effects. The largest gains under \concr occur for 2LS on non-termination Integer Underflow/Overflow (\base 1 $\rightarrow$ \concr 10, +9), CPAchecker on termination Bit Calculation (+8), and CPAchecker on termination Integer Underflow/Overflow (+6). \sliceconcr matches \concr on these cases, showing that adding \slice is not consistently additive. In contrast, \slice alone ranges from $-4$ (UAutomizer on non-termination Recursion) to +3 (CPAchecker on non-termination Integer Underflow/Overflow). This suggests that the largest gains are mainly driven by input-scenario specialization: \concr exposes specific behaviors without requiring reasoning over all nondeterminism, whereas \slice mainly removes surrounding context and has smaller, analyzer- and feature-dependent effects.

\begin{rqtakeaway}
\textbf{RQ4 Takeaway.}
Preprocessing effects are strongly feature-sensitive. Arithmetic and bit-level categories benefit most, especially under configurations containing \concr, while recursion and data-structure benchmarks show limited gains. The responsive features still vary across analyzers and verification tasks.
\end{rqtakeaway}

\subsection{RQ5: Incorrect Classification Reduction}
Our results show that preprocessing does not consistently reduce incorrect classifications relative to \base; its effect varies across analyzers, benchmark classes, and configurations. Figure~\ref{figure:rq1} shows that, for termination, incorrect classifications decrease for Proton (\base 5 $\rightarrow$ \slice 4, \concr/\sliceconcr 3), UAutomizer (\base 1 $\rightarrow$ \concr/\sliceconcr 0), and CPAchecker (\base 5 $\rightarrow$ \concr/\sliceconcr 1), although \slice increases CPAchecker incorrect classifications from 5 to 7. For non-termination, reductions are limited to \slice for Proton (5 $\rightarrow$ 4) and CPAchecker (1 $\rightarrow$ 0); by contrast, \concr and \sliceconcr increase incorrect counts for several analyzers, including Proton (5 $\rightarrow$ 9/8), UAutomizer (2 $\rightarrow$ 7/7), AProVE (4 $\rightarrow$ 8/8), and CPAchecker (1 $\rightarrow$ 4/3).

Table~\ref{table:rq2} explains why incorrect counts can both decrease and increase. For termination, Proton's 5 \base-\emph{wrong} cases are partly removed: \slice turns one into \emph{unknown}, \concr solves two, and \sliceconcr solves one while turning another into \emph{unknown}. Thus, preprocessing can reduce incorrect outcomes either by enabling a correct result or by avoiding an incorrect conclusive one. However, it can also create new incorrect outcomes. For CPAchecker with \slice, the 5 \base-\emph{wrong} termination cases remain \emph{wrong}, while 2 \base-\emph{unknown} cases become \emph{wrong}. This means preprocessing makes additional benchmarks conclusive, but some are classified incorrectly. For Proton with \concr on non-termination, the 5 \base-\emph{wrong} cases remain \emph{wrong}, and 4 new incorrect cases appear, 3 of which were previously correct under \base. Thus, even semantics-preserving transformations can change how backend analyzers respond to verification tasks.

Incorrect classifications require particular attention. Ideally, semantics-preserving preprocessing should preserve or improve correctness, and analyzer performance should improve or at least not deteriorate, since slicing and concretization are intended to simplify verification of transformed instances. However, our results show that incorrect classifications and performance deterioration remain challenges for several tools. This aligns with prior studies showing that (non-)termination analyzers often perform well on theoretical benchmarks such as TermCOMP and SV-COMP, but struggle more with real-world-driven C programs involving complex features and low-level semantics. The observed variation across configurations is informative: transformed benchmarks can expose backend analyzer limitations, and these mixed outcomes provide useful evidence for improving (non-)termination tools. For application programmers, the results indicate when preprocessing should be used cautiously and interpreted as part of a portfolio workflow. For tool developers, they identify analyzer--feature and analyzer--transformation combinations where additional engineering is needed to improve robustness beyond increasing solved counts.

Figure~\ref{figure:rq1} further shows that configurations differ in the magnitude and direction of incorrect-classification changes. \slice produces only small changes, mostly 0 or $\pm 1$, with one larger increase for CPAchecker on termination (+2). In contrast, \concr has the strongest changes in both directions: it yields the largest reductions for termination, especially CPAchecker ($-4$) and Proton ($-2$), but also the largest increases for non-termination, especially UAutomizer (+5), Proton (+4), and AProVE (+4). \sliceconcr largely mirrors \concr: it preserves the same termination reductions, especially for CPAchecker ($-4$) and Proton ($-2$), and similar non-termination increases, especially for UAutomizer (+5), AProVE (+4), and Proton (+3). This suggests that \slice is relatively conservative, whereas \concr-based configurations have stronger mixed effects: input-scenario specialization can make additional variants conclusive, but not necessarily correctly classified.

Finally, Table~\ref{table:rq4} shows that feature-level changes in incorrect classifications are concentrated in specific analyzer--feature combinations rather than distributed uniformly across features. Reductions occur mainly in Integer Underflow/Overflow and Bit Calculation. For Integer Underflow/Overflow, incorrect classifications decrease on non-termination for Proton ($-2$ under \slice, $-4$ under \sliceconcr), AProVE ($-1$ under \concr, $-3$ under \sliceconcr), and CPAchecker ($-1$ under \slice), and on termination for Proton under \slice and \sliceconcr and for UAutomizer and CPAchecker under \concr and \sliceconcr. In contrast, increases concentrate mainly in non-termination Array and Pointer Manipulation benchmarks. For Array, incorrect classifications increase under \concr/\sliceconcr for Proton, UAutomizer, and AProVE by +1/+3, +3/+4, and +1/+2, respectively, and under \slice/\sliceconcr for 2LS by +1/+1. This suggests that preprocessing helps when the remaining task matches an analyzer's strengths, whereas Array and Pointer Manipulation cases may become conclusive without fully resolving the underlying memory- and aliasing-related difficulty.

\begin{table*}[!t]
\centering
\footnotesize
\setlength{\tabcolsep}{7pt}
\caption{Efficiency and runtime impact under preprocessing configurations. Entries report the average Solved Ratio (SR), average Total Variant Time (TVT), average Average Variant Time (AVT), and median Median Variant Time (MVT) for non-terminating (NT) and terminating (T) benchmarks across analyzers and configurations. Runtime statistics are computed over correctly solved programs. Results for \concr and \sliceconcr are evaluated over selected input-scenario variants, not over all nondeterministic executions of the original program.}
\label{table:rq6}
\begin{tabular}{l|l|cc|cc|cc|cc|cc|cc}
\noalign{\hrule height 1pt}
\multicolumn{1}{c|}{\multirow{2}{*}{\shortstack[c]{\textbf{Metric}}}} & 
\multicolumn{1}{c|}{\multirow{2}{*}{\shortstack[c]{\textbf{Configuration}}}} & 
\multicolumn{2}{c|}{\textbf{Athena}} & 
\multicolumn{2}{c|}{\textbf{Proton}} & 
\multicolumn{2}{c|}{\textbf{UAutomizer}} & 
\multicolumn{2}{c|}{\textbf{AProVE}} & 
\multicolumn{2}{c|}{\textbf{CPAchecker}} & 
\multicolumn{2}{c}{\textbf{2LS}} \\
\cline{3-14}
& & NT & T & NT & T & NT & T & NT & T & NT & T & NT & T \\
\noalign{\hrule height 1pt}
\multirow{4}{*}{Avg. SR (\%)}
& \base & 69.64 & 59.02 & 53.57 & 85.25 & 41.07 & 47.54 & 23.21 & 40.98 & 35.71 & 16.39 & 1.79 & 47.54 \\
& \slice & 63.10 & 59.15 & 62.65 & 88.39 & 39.14 & 46.04 & 28.42 & 40.30 & 38.69 & 15.85 & 6.70 & 51.23 \\
& \concr & 59.64 & 69.34 & 64.11 & 90.82 & 48.93 & 57.54 & 41.96 & 53.44 & 41.79 & 46.39 & 39.64 & 57.38 \\
& \sliceconcr & 57.32 & 66.04 & 67.38 & 90.77 & 43.10 & 53.25 & 45.60 & 50.87 & 42.53 & 45.52 & 42.77 & 61.37 \\
\hline
\multirow{4}{*}{Avg. TVT (s)}
& \base & 14.27 & 8.11 & 3.37 & 76.71 & 4.68 & 3.90 & 2.11 & 3.68 & 2.14 & 18.71 & 0.01 & 0.20 \\
& \slice & 14.19 & 8.56 & 12.07 & 119.89 & 7.35 & 7.07 & 3.33 & 5.61 & 2.92 & 1.83 & 0.06 & 0.34 \\
& \concr & 29.33 & 32.54 & 63.14 & 90.72 & 23.44 & 50.42 & 20.58 & 27.88 & 41.97 & 98.40 & 1.47 & 2.35 \\
& \sliceconcr & 35.49 & 57.22 & 80.76 & 129.53 & 30.90 & 64.95 & 32.44 & 35.54 & 43.32 & 76.95 & 2.09 & 3.67 \\
\hline
\multirow{4}{*}{Avg. AVT (s)}
& \base & 14.27 & 8.11 & 3.37 & 76.71 & 4.68 & 3.90 & 2.11 & 3.68 & 2.14 & 18.71 & 0.01 & 0.20 \\
& \slice & 10.88 & 5.22 & 6.80 & 73.04 & 4.11 & 3.38 & 2.71 & 3.69 & 2.28 & 1.31 & 0.04 & 0.22 \\
& \concr & 3.57 & 4.25 & 8.88 & 10.92 & 4.17 & 9.86 & 4.26 & 4.29 & 5.07 & 26.34 & 0.24 & 0.29 \\
& \sliceconcr & 3.54 & 4.74 & 6.52 & 9.44 & 3.19 & 7.80 & 3.45 & 2.82 & 4.75 & 7.75 & 0.27 & 0.31 \\
\hline
\multirow{4}{*}{Med. MVT (s)}
& \base & 1.97 & 1.87 & 2.03 & 43.70 & 0.00 & 0.00 & 0.00 & 0.00 & 0.00 & 0.00 & 0.00 & 0.00 \\
& \slice & 1.88 & 1.87 & 2.08 & 47.24 & 0.00 & 2.41 & 0.00 & 0.00 & 0.00 & 0.00 & 0.00 & 0.29 \\
& \concr & 1.72 & 1.73 & 6.07 & 7.43 & 3.00 & 2.76 & 3.16 & 2.94 & 4.86 & 4.85 & 0.36 & 0.37 \\
& \sliceconcr & 1.75 & 1.77 & 5.60 & 7.82 & 2.64 & 2.73 & 3.23 & 2.81 & 4.86 & 4.83 & 0.36 & 0.37 \\
\noalign{\hrule height 1pt}
\end{tabular}
\end{table*}

\begin{rqtakeaway}
\textbf{RQ5 Takeaway.}
Preprocessing does not consistently reduce incorrect classifications. \slice has relatively conservative effects, whereas \concr-based configurations produce larger but mixed changes, with feature-level effects concentrated in specific analyzer--feature pairs.
\end{rqtakeaway}

\subsection{RQ6: Efficiency and Runtime Impact}
Table~\ref{table:rq6} compares runtime efficiency across preprocessing configurations. 
Because TVT, AVT, and MVT are computed over correctly solved programs, they should be interpreted together with SR: lower runtime may reflect a smaller or easier solved subset, while higher runtime may reflect additional solved cases or generated variants. Configurations containing \concr substantially increase total analysis cost (TVT), whereas \slice usually has a smaller effect. For example, for Athena on termination, TVT increases from 8.11s under \base to 32.54s under \concr and 57.22s under \sliceconcr, while \slice changes it only slightly to 8.56s. This is expected because TVT accumulates runtime across all analyzed variants, and \concr-based configurations introduce additional input-scenario variants.

Despite increasing total runtime, preprocessing can reduce per-variant runtime (AVT), depending on the analyzer and benchmark class. For Athena on termination, AVT decreases from 8.11s under \base to 5.22s under \slice, 4.25s under \concr, and 4.74s under \sliceconcr. In contrast, Proton's non-termination AVT increases from 3.37s under \base to 6.80s under \slice, 8.88s under \concr, and 6.52s under \sliceconcr. Since AVT normalizes by the number of analyzed variants, these results show that preprocessing can make individual variants easier to analyze, but not uniformly across analyzers and benchmark classes.

Finally, configurations containing \concr generally affect typical-case runtime (MVT) more than \slice, though the direction depends on the analyzer and benchmark class. For Proton's termination benchmarks, MVT increases slightly from 43.70s under \base to 47.24s under \slice, but drops to 7.43s under \concr and 7.82s under \sliceconcr. For CPAchecker's non-termination benchmarks, MVT remains unchanged under \slice but increases to 4.86s under both \concr and \sliceconcr. Because MVT reflects a typical analyzed variant and is less sensitive to outliers than AVT, \concr-based configurations alter typical-variant difficulty more strongly than \slice, with effects that may be positive or negative.

\begin{rqtakeaway}
\textbf{RQ6 Takeaway.}
Preprocessing increases total analysis cost mainly by generating additional variants, but it does not uniformly make those variants harder to analyze. Per-variant runtime effects vary across analyzers and benchmark classes.
\end{rqtakeaway}

\subsection{RQ7: Tractability--Generality Trade-offs}
The results show a trade-off among detectability, solver cost, and semantic scope. Configurations containing \concr can improve detectability for selected analyzers, but establish correctness or divergence only over selected input-scenario variants, not all nondeterministic executions. In contrast, \slice preserves broader loop-relevant behavior and program semantics, but usually yields smaller detectability gains; its main benefit is diagnostic localization of unresolved behavior to specific loops. The tractability cost also differs: \slice has limited overhead, whereas \concr and \sliceconcr can increase TVT by generating additional variants. However, AVT and MVT show that some variants become easier to analyze, so higher TVT does not necessarily imply worse per-variant tractability.

\begin{rqtakeaway}
\textbf{RQ7 Takeaway.}
\concr improves detectability by specializing inputs, but narrows semantic scope and can increase total cost. \slice preserves broader behavior and mainly helps localization, while \sliceconcr combines both effects without consistently dominating either.
\end{rqtakeaway}

\subsection{RQ8: Integration Potential}
In principle, sound preprocessing should improve analyzer results; mixed outcomes suggest limitations in how backend analyzers handle transformed tasks. Our results support integrating \slice, \concr, and \sliceconcr into \base verification as complementary options rather than replacements. Across RQ1--RQ5, no configuration is uniformly best: each may improve, degrade, or leave results unchanged depending on the analyzer, benchmark class, and program features. RQ2 further shows that preprocessing solves different benchmark subsets from \base. Thus, \slice supports structural simplification and loop-level localization, \concr supports input-scenario specialization, and \sliceconcr may help when both effects are relevant. 
The behavior of \concr and \sliceconcr also suggests potential for hybrid static--dynamic workflows, with effects depending on the analyzer, task, and program features. Accordingly, integration should be adaptive: retain \base as the default reference, then selectively apply \slice for difficult loop structures, \concr when nondeterministic input scenarios limit analysis, and \sliceconcr when both effects may help. Results should be interpreted in a portfolio style.

\begin{rqtakeaway}
\textbf{RQ8 Takeaway.}
Preprocessing should be integrated as an adaptive portfolio. \base should remain the default reference, while \slice, \concr, and \sliceconcr are applied selectively based on the likely source of analysis difficulty.
\end{rqtakeaway}

\section{Discussion}
\label{section:discussion}
This section discusses and summarizes the main implications of the empirical results.

\paragraph{Preprocessing as analyzer-, task-, and feature-dependent refinement}
The empirical results show that preprocessing acts as an analyzer-, task-, and feature-dependent refinement rather than a uniformly beneficial optimization. Across analyzers and verification tasks, preprocessing changes the set of cases each analyzer handles instead of monotonically improving on \base, with effects generally more visible for non-termination than for full termination resolution. Outcomes also vary across features: arithmetic and bit-level benchmarks benefit most often, whereas recursion and data-structure benchmarks show more limited gains. This variation helps application programmers interpret preprocessing outcomes cautiously and provides tool developers with evidence of analyzer--feature and analyzer--transformation interactions that require stronger support.

\paragraph{Configuration-level effects and trade-offs}
\slice, \concr, and \sliceconcr modify different dimensions of the verification task. \slice removes termination-irrelevant context and isolates loop-level obligations; \concr reduces nondeterministic input variability by analyzing selected input-scenario variants; and \sliceconcr combines structural isolation with input specialization. Their effects differ in magnitude and direction: \slice is generally the most conservative, improving loop-level localization and producing modest, mixed changes in solved cases and incorrect classifications. \concr is more aggressive, yielding the largest gains in some analyzer--task pairs and reducing certain incorrect classifications, but it can also lose solved cases or increase incorrect classifications when specialization changes the behaviors exposed to the analyzer. \sliceconcr can help when both structural coupling and nondeterminism hinder analysis, but its benefits are not consistently additive beyond \concr. The configurations also differ in generality and cost: \slice preserves broader loop-relevant behavior, whereas \concr and \sliceconcr narrow semantic scope to selected variants; similarly, \slice usually has moderate cost, while \concr-based configurations may improve per-variant tractability but increase total runtime by generating multiple variants.

\paragraph{Toward adaptive integration}
These results suggest that preprocessing is most useful not as a replacement for \base, but as a complementary way to obtain alternative views of difficult verification cases. When original analysis is inconclusive, transformed variants can help indicate whether the difficulty is associated with loop structure, input variability, analyzer limitations, or feature-specific reasoning gaps. Overall, the results support an adaptive workflow in which preprocessing complements, diagnoses, and supports original-program analysis.

\section{Threats to Validity}
\label{section:threats}
This section discusses the main threats to validity and mitigation steps.

\paragraph{Construct validity}
A main threat concerns how concretized variants represent the original program. Input-driven concretization depends on the input-generation strategy, number of generated inputs, and behavioral diversity of those inputs; different choices may expose different execution domains and affect quantitative results. Since our goal is to study verification behavior under input-scenario specialization, not input generation, concretized variants should be interpreted as selected input-scenario instances: they preserve behavior for exercised inputs but do not establish (non-)termination for all nondeterministic executions of the original program.

\paragraph{Internal validity}
Generated variants and their ground-truth labels require controlled validation to preserve semantic consistency. We mitigate this by assigning labels according to each generated program instance and applying a consistent aggregation policy across configurations. However, validation may limit scalability to larger benchmark sets, and alternative aggregation policies could yield different program-level summaries. In addition, concretization benefits depend on whether backend analyzers exploit fixed input values. Because our pipeline does not add analyzer-specific instrumentation or extra constant-propagation passes, these differences remain part of the measured analyzer behavior.

\paragraph{Performance validity}
Configurations that generate multiple variants may introduce opportunity bias relative to \base, since verification effort is distributed across several decomposed analyses rather than one original-program analysis. Efficiency conclusions also depend on the runtime model: TVT captures cumulative runtime across variants, while AVT and MVT capture per-variant and typical-case costs over correctly solved programs. We therefore report TVT, AVT, MVT, and SR together to relate runtime cost to the solved subset used for timing analysis. Conclusions may still vary with timeout budgets, hardware, tool settings, or variant-generation policies.

\paragraph{External validity}
Generality is limited by the selected benchmark suite and analyzers. Although the benchmarks are real-world-driven and include challenging C/C++ features, and the analyzers cover diverse verification paradigms, the observed effects may differ for other suites, larger codebases, semantic models, language features, or (non-)termination tools.


\section{Conclusion and Future Work}
\label{section:conclusion}
We presented \tool, a lightweight, tool-independent preprocessing front end for C programs that applies loop-based slicing and input-driven concretization before (non-)termination analysis. Our empirical study shows that preprocessing is not uniformly beneficial: its effects depend on the analyzer, verification task, and program features. \slice provides conservative structural isolation and loop-level localization, whereas \concr can improve detectability for selected input-scenario variants at the cost of narrower semantic scope and additional analysis effort. The combined \sliceconcr configuration is not consistently additive. Overall, the results support using preprocessing adaptively as a complement to \base, providing alternative verification views that help diagnose difficult cases and guide future analyzer engineering.

Our findings suggest several directions for future work. First, adaptive preprocessing strategies could decide when and how to apply slicing or concretization, guided by structural features or early analysis feedback. Such strategies could preserve preprocessing benefits while avoiding unnecessary variant generation and analyzer-sensitive regressions. Second, input discovery for concretization could be improved through lightweight search, heuristic exploration, or coverage-guided generation to expose informative behaviors. Third, future work could extend the loop-focused slicing criteria and preservation arguments to recursive call structures. Fourth, the empirical scope could be broadened by evaluating additional analyzers and benchmark suites, including larger real-world codebases with richer C/C++ features. Finally, future work could study whether these source-level transformations extend beyond termination analysis to tasks such as safety verification, liveness reasoning, or resource-bound inference.

\section*{Data Availability}
\label{section:data-availability}
The artifacts supporting this paper are publicly available at \url{https://github.com/negarfathi/FocusTNT}, including the implementation, benchmarks, evaluation results, and reproduction scripts.

\section*{Acknowledgments}
\label{section:acknowledgments}
The authors used OpenAI ChatGPT (GPT-4o)~\cite{chatgpt} solely for grammar checking, language editing, and clarity improvement. All of these edits were reviewed and validated by the authors, who take full responsibility for the final manuscript.

\bibliographystyle{IEEEtran}
\bibliography{references}

\end{document}